\documentclass[lineno]{biometrika}

\usepackage{amsmath}

\usepackage{newtxtext}
\usepackage[subscriptcorrection]{newtxmath}
\usepackage[dvipsnames]{xcolor}

\graphicspath{{./art/}}

\usepackage[plain,noend, linesnumbered]{algorithm2e}

\makeatletter
\renewcommand{\algocf@captiontext}[2]{#1\algocf@typo. \AlCapFnt{}#2} 
\def\@algocf@capt@plain{top}
\renewcommand{\algocf@makecaption}[2]{%
  \addtolength{\hsize}{\algomargin}%
  \sbox\@tempboxa{\algocf@captiontext{#1}{#2}}%
  \ifdim\wd\@tempboxa >\hsize
    \hskip .5\algomargin%
    \parbox[t]{\hsize}{\algocf@captiontext{#1}{#2}}
  \else%
    \global\@minipagefalse%
    \hbox to\hsize{\box\@tempboxa}
  \fi%
  \addtolength{\hsize}{-\algomargin}%
}
\makeatother

\usepackage{hyperref}

\makeatletter
\def\lstAZ{A, B, C, D, E, F, G, H, I, J, K, L, M, N, O, P, Q, R, S, T, U, V, W, X, Y, Z}
\def\lstaz{a, b, c, d, e, f, g, h, i, j, k, l, m, n, o, p, q, r, s, t, u, v, w, x, y, z}
\def\lstAZBB{B, C, D, E, F, G, H, I, J, K, L, M, N, O, Q, R, T, U, V, W, X, Y, Z}
\newcommand{\MkScr}[1]{\expandafter\def\csname s#1\endcsname{\mathscr{#1}}}
\newcommand{\MkUp}[1]{\expandafter\def\csname u#1\endcsname{\mathrm{#1}}}
\newcommand{\MkFrak}[1]{\expandafter\def\csname f#1\endcsname{\mathfrak{#1}}}
\newcommand{\MkCal}[1]{\expandafter\def\csname c#1\endcsname{\mathcal{#1}}}
\newcommand{\MkBB}[1]{\expandafter\def\csname #1#1\endcsname{\mathbb{#1}}}

\@for\i:=\lstAZ\do{%
	\expandafter\MkScr \i  %
	\expandafter\MkFrak \i  %
	\expandafter\MkUp \i %
	\expandafter\MkCal \i  %
		  }    
\@for\i:=\lstaz\do{%
	\expandafter\MkUp \i   }    
	
\@for\i:=\lstAZBB\do{%
	\expandafter\MkBB \i     }
\makeatother
\newcommand{\PP}{\mathbb{P}}
\newcommand{\dd}{\mathrm{d}}
\newcommand{\ind}{\mathbf{1}}

\graphicspath{{Fig/}}

\begin{document}

\jname{Biometrika}
\jyear{2025}
\jvol{112}
\jnum{1}
\cyear{2025}
\accessdate{Advance Access publication on 14 February 2025}


\markboth{S. Grazzi \& G. Zanella}{Parallel computations for  Metropolis Markov chains with Picard maps}

\title{Parallel computations for  Metropolis Markov chains with Picard maps}

\author{S. Grazzi}
\affil{Department of Decision Sciences and BIDSA, Bocconi University\\ Via Roentgen 1, 20136, Milan, Italy
\email{sebastiano.grazzi@unibocconi.it}}

\author{G. Zanella}
\affil{Department of Decision Sciences and BIDSA, Bocconi University\\ Via Roentgen 1, 20136, Milan, Italy
\email{giacomo.zanella@unibocconi.it}}

\maketitle
\thispagestyle{empty} 

\begin{abstract}
We develop parallel algorithms for simulating zeroth-order (aka gradient-free) Metropolis Markov chains based on the Picard map.  For Random Walk Metropolis Markov chains targeting log-concave distributions $\pi$ on $\mathbb{R}^d$, our algorithm  generates samples close to $\pi$ in $\cO(\sqrt{d})$ parallel iterations with $\cO(\sqrt{d})$ processors, therefore speeding up the convergence of the corresponding sequential implementation by a factor $\sqrt{d}$. Furthermore, a modification of our algorithm generates samples from an approximate measure $ \pi_r$ in $\cO(1)$ parallel iterations and $\cO(d)$ processors. We empirically assess the performance of the proposed algorithms in high-dimensional regression problems, an epidemic model where the gradient is  unavailable and a real-word application in precision medicine. Our algorithms are straightforward to implement and may constitute a useful tool for practitioners seeking to  sample from a prescribed distribution $\pi$ using only point-wise evaluations of $\log\pi$ and parallel computing.
\end{abstract}


\section{Introduction}
\subsection{Zeroth-order Markov chain Monte Carlo}
Markov chain Monte Carlo (MCMC) methods are the computational workhorse of many scientific areas. These methods consist of simulating a Markov chain $(X_0,X_1,\dots)$ whose limiting measure coincides with a probability distribution of interest $\pi$, such as a posterior distribution in Bayesian statistics. Then, ergodic averages $\frac{1}{N} \sum_{i=1}^N g(X_i)$ can be used to estimate expectations $ \int g(x)\pi(\dd x)$ for any integrable function $g$, see e.g.\ \cite{roberts2004general}.

In this paper we focus on \emph{zeroth-order} (aka gradient-free or derivative-free) MCMC methods, which only require point-wise evaluations of $\log\pi$ up to an additive unknown constant, as opposed to first-order methods, which require computing the gradient of $\log \pi$. 
Zeroth-order methods are commonly used when gradient information is unavailable, for either practical, computational, or mathematical reasons \citep[cf. ][Chapter 1]{conn2009introduction}.
This may occur for example in Bayesian applications with black-box likelihood evaluations   (e.g. in form of proprietary code, code written by domain experts, or models requiring complex numerical solvers) or when the gradient is not well-defined such as in models with censored data (e.g. the 
SIR model considered in Section~\ref{sec: numerical illustrations}), in Pseudo-Marginal MCMC \citep{508a6095-c33f-3d5d-9a9a-f004c0c0a59a} and in Approximate Bayesian Computation (ABC) MCMC setting \citep{marjoram2003markov}. 
See e.g.\ \citet[][]{conn2009introduction, carrillo2022consensus, pavliotis2022derivative, bourabee2025nounderrunsamplerlocallyadaptivegradientfree, grumitt2024flow} and other references on zeroth-order sampling and optimization for more discussion about applications of gradient-free methods.

For gradient-free sampling algorithms targeting log-concave distributions, the best complexity result w.r.t.\ the number of dimensions $d$ of the target $\pi$ is $\cO(d)$, and it is achieved by classical Random Walk Metropolis  \citep{dwivedi2018log, andrieu2024explicit}. Here and in what follows, the $\cO$ notation ignores constants and logarithmic terms with respect to $d$ and \emph{complexity} refers to the number of point-wise evaluations of $\log\pi$ required by the algorithm in order to obtain one sample from $\pi$ up to a fixed error. This complexity matches the one of gradient-free convex optimization \citep{nesterov2017random}.

\subsection{Parallel sampling}
In this paper, we are interested in parallelization strategies for MCMC algorithms. This is particularly interesting given the increasing availability of parallel computing architectures -- such as clusters of CPUs and Graphics Processing Units (GPUs)  --  which play an increasing central role in computation-based sciences. 
In the MCMC context, the simplest approach to parallelize computations is to run multiple independent chains. 
This is straightforward to implement in parallel but does not reduce the convergence period (also called burn-in or warm-up) of each chain \citep[][]{rosenthal2000parallel}. 
Alternatively, one can consider algorithms which utilize $K > 1$ parallel processors where, at each iteration, each processor evaluates independently and simultaneously $\log \pi$ for different parameter values and use the resulting information to speedup convergence to stationarity. 
Examples include \emph{pre-fetching} methods \citep{brockwell2006parallel}, which compute $\log \pi$ at each potential state of the Markov chain for $j \ge 1$ steps ahead, and \emph{Multiple-try} \citep{bookFrenkel, multipletry} which simulates $K$ proposal states at each iteration and computes $\log \pi$ in each proposed state before deciding which state to move to. Unfortunately, at least for zeroth-order methods and log-concave targets, 
these methods achieve only a $\cO(\log(K))$ speedup factor, as recently proven in \citep{pozza2024fundamental}.

An alternative strategy to parallelize the simulation of Markov chains is provided by the \emph{Picard recursion} detailed in Section~\ref{sec: Picard map} below, where the simulation of a Markov chain is reformulated as a \emph{fixed-point problem} over trajectories. Methods based on the Picard recursion have recently witnessed a growing interest, particularly for first-order (i.e.\ gradient-based) Markov chains for both MCMC methods \citep{yu2024parallelized,anari2024fast} and generative models \citep{shih2023parallel, zhou2024parallel, NEURIPS2024_f162fa05, gupta2024faster}. However, existing methods and analyses focus on first-order unadjusted settings where the function  $f$  defined in \eqref{eq: 1} below is smooth with respect to its first argument, which is not the case for Metropolis-Hastings algorithms -- see \eqref{eq:zero_f}. Our work avoids the classical smoothness conditions on $f$ and instead focuses on zeroth-order Metropolis-type Markov chains, where $f$ is piecewise constant with respect to its first argument. This setting leads to a substantially different analysis and convergence results for the induced Picard map, as highlighted in  Section~\ref{sec: technical comparison picard maps} below.
\subsection{Structure of the paper and main results}

Section~\ref{sec: methodology} introduces the Picard map and our proposed \emph{Online Picard algorithm}. Section~\ref{sec: Complexity of Picard algorithms} presents its theoretical analysis. In particular, Corollary~\ref{corollary: complexity rwm}  shows that, given a log-concave target density $\pi$ on $\RR^d$ and $K$ processors, with $K$ up to  $\mathcal{O}(\sqrt{d})$, the Online Picard algorithm applied to a Random Walk Metropolis Markov chain obtains a sample from $\pi$ in $\cO(d/K)$ iterations, thus accelerating zeroth-order sequential algorithms by a factor of $K$ and achieving optimal speedup (i.e.\ linear in $K$).
To the best of our knowledge, this is the first parallel zeroth-order MCMC scheme with provably linear speedup in the canonical log-concave set-up.
Section~\ref{sec: Metropolis within gibbs} extends the convergence results to Picard algorithms applied to Metropolis within Gibbs Markov chains, which we empirically observe to perform better. 
Section~\ref{sec: approximate online-picard} proposes an \emph{Approximate Online Picard algorithm} that scales well also when $K\gg \sqrt{d}$, at the price of introducing a bias in the invariant distribution, whose size is assessed through numerical simulations. 
For both the approximate and exact schemes, the empirical performance observed in Section~\ref{sec: numerical illustrations} closely matches the theory developed and shows substantial speedups (up to over 100 times) relative to the corresponding sequential algorithms. All proofs are provided in the appendix.

These results, combined with the simplicity of the proposed algorithms, offer practitioners promising directions to parallelize computations for Bayesian problems with black-box, expensive models and no access to gradient information.

\section{Sampling algorithms based on Picard maps}\label{sec: methodology}
\subsection{Picard map for Markov chain simulation}\label{sec: Picard map}
Consider a discrete-time Markov chain $(X_i)_{i\in \NN}$ on $\cX \subseteq  \RR^d$ defined as 
\begin{align}
    \label{eq: 1}
X_{i+1} 
&= X_{i} + f(X_{i}, W_{i})\,, 
&i=0,1,\dots
\end{align}
for some initial position $X_0\in\cX$, measurable function $f\colon \cX\times\mathcal{W} \to \RR^d$ and random innovations $W_0,W_1,\dots \stackrel{\text{i.i.d.}}\sim \nu$ on $\mathcal{W}$.
The canonical way to simulate the first $K$ steps of the chain, $X=(X_0,\dots,X_K)$, is to go through the recursion in \eqref{eq: 1} for $i=0,1,\dots,K-1$.
 This procedure requires $K$ sequential calls to the function $f$. 
Alternatively, one can define  $X$  through the set of equations
\begin{align*}
X_{i+1} &= X_0 + \sum_{\ell=0}^{i}  f(X_\ell, W_\ell)\,,
&i=0,1,\dots\,.
\end{align*}
Following this representation, one can define the Picard map
(named after the celebrated \emph{Picard–Lindel\"of theorem}), $\Phi\colon \cX^{K+1}\times\mathcal{W}^K\mapsto\cX^{K+1}$, which takes as input a trajectory $X\in\cX^{K+1}$ and the innovations $W=(W_0,\dots,W_{K-1})\in\mathcal{W}^K$, and returns a new trajectory $X'=(X'_0,\dots,X'_K) =  \Phi(X,W)$ defined as
\begin{align}
    \label{eq: picard recursion 0}
    X_{i}' 
    &= \Phi_i(X,W) =
    \begin{cases}
        X_0 & i = 0 \\
        X_0 + \sum_{\ell=0}^{i-1} f(X_\ell, W_\ell) & 0<i\leq K\,.
    \end{cases}
\end{align}
Given $W \in \cW^K$, the map $\Phi$ is deterministic and its fixed point, i.e.\ the trajectory $X \in\cX^{K+1}$ satisfying $X=\Phi(X,W)$, is unique and coincides with the solution to \eqref{eq: 1}. 

This motivates computing $X$ as the limit of the recursion 
$X^{(j)}=\Phi(X^{(j-1)},W)$ for $j=1,2,\dots$, with  initialization $X^{(0)} \in \cX^{K+1}$. See Figure~\ref{fig:picard-illustration} for an illustration. A canonical choice for initialization is the constant trajectory $X^{(0)}=(X_0,\dots,X_0)\in\cX^{K+1}$.
The computational advantage of the Picard recursion over the classical one defined in \eqref{eq: 1} is that the $K$ calls to the function $f$ in \eqref{eq: picard recursion 0} can be executed in parallel. 
This can be useful when computing $f$ is expensive.

The Picard recursion is guaranteed to converge to its fixed point in at most $K$ iterations because, by construction, $X^{(j)}$ coincides with the solution of \eqref{eq: 1} in the first $j$ coordinates. 
However, since the sequential recursion in \eqref{eq: 1} also requires $K$ iterations to compute $X$, relevant applications of the Picard recursion are those where $X^{(j)}$ gets close (or exactly equal) to its fixed point for $j \ll K$ iterations, thus exploiting parallel computing to shorten the overall runtime.

\begin{figure}
    \centering
    \includegraphics[width=0.99\linewidth]{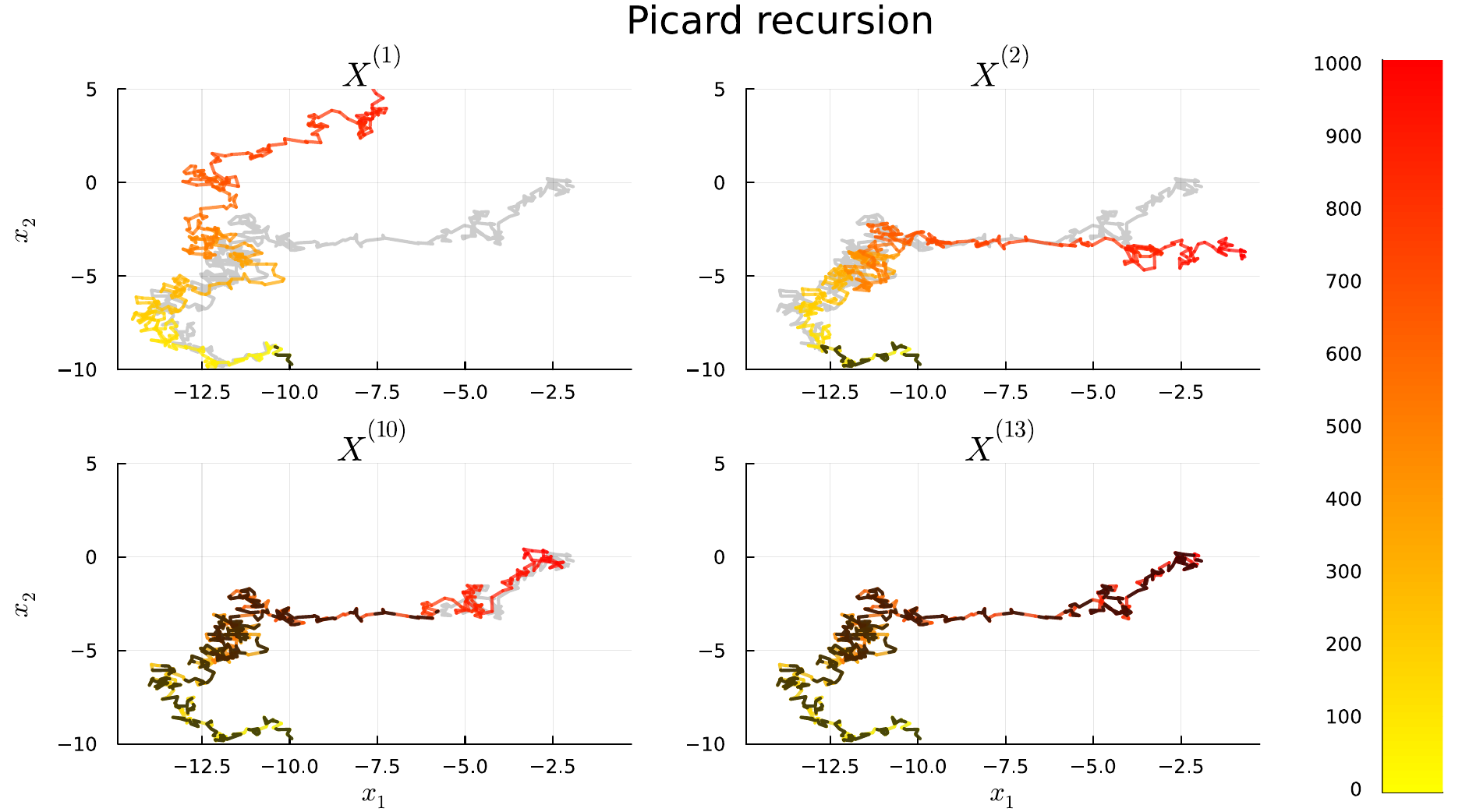}
    \caption{Traces on the $(x_1, x_2)$-plane of the Picard recursion $X^{(j)}$  for $K = 1000$ and $j = 1,2,10,13$ (top left - bottom right). 
    The time-index of each sequence is shown with a yellow-red
    gradient color. The underlying Markov chain is a $d = 100$ dimensional Random Walk Metropolis with stepsize $\xi = 2/\sqrt{d}$ targeting a standard Gaussian distribution. The gray line is the fixed point (i.e. the output of the sequential algorithm). The dashed line corresponds to the part of the trajectory that has converged to its fixed point.
    }
    \label{fig:picard-illustration}
\end{figure}

\subsection{Picard algorithms with piecewise constant map $\Phi$}\label{sec: classical Picard algorithm}

The Picard map applied to standard zeroth-order Markov chains (such as those described in Section \ref{subsec: mh kernels}) differs from those of classical applications because, for every noise variable $W\in \cW^{K}$, \emph{the Picard map $X \mapsto \Phi(X, W)$ is piecewise constant.}
Other relevant examples satisfying this condition include Markov chains defined on discrete spaces. 
A piecewise constant map $\Phi$ has many  implications. For example, it implies that the fixed point of $\Phi$ can be reached \emph{exactly} after $j < K$ parallel iterations. When the fixed point is reached, the output of the Picard recursion matches that of the sequential counterpart, with no additional bias induced by the Picard recursion. 
This is in contrast to classical applications of Picard recursions, where stopping the Picard recursion before $K$ parallel iterations introduces bias and requires tuning of an additional error parameter. 
In what follows, we exploit this property to develop parallel algorithms that allocate computational resources more efficiently. 

As above, assume we seek to simulate $N$ steps of the Markov chain defined in \eqref{eq: 1}, $(X_0,\dots,X_N)$, using $K\le N$ independent processors. In the classical Picard algorithm, one updates the first $K$ steps, $(X_1,\dots,X_K)$, with the Picard recursion in \eqref{eq: picard recursion 0} until a stopping criterion is met, then proceeds with the next $K$ ones, $(X_{K+1},\dots,X_{2K})$, and so on until reaching $X_N$. See Algorithm~\ref{alg: PA} for the associated pseudocode in the case of piecewise constant maps.

Alternatively, one can monitor the coordinates that have already reached their fixed point while the algorithm runs, and avoid wasting computational power in updating them further, directly allocating available processors to subsequent coordinates.
This leads to the following \emph{Online Picard algorithm}:
given a sequence of noise $W=(W_0,W_1,\dots)\in \cW^{\infty}$ and initialization $X^{(0)} \in \cX^{\infty}$, define the $j$-th parallel recursion of the Online Picard algorithm mapping $X^{(j)}\in\cX^{\infty}$ to $X^{(j+1)}\in\cX^{\infty}$
as
\begin{align}
\label{eq: online picard 1}
X_{i}^{(j+1)} &= 
\begin{cases}
X_{i}^{(j)} &  i \le L^{(j)},\\
 \Phi_{i - L^{(j)}}\left(X_{L^{(j)}:U^{(j)}}^{(j)}, W_{L^{(j)}:U^{(j)}-1}\right) & L^{(j)}< i \le U^{(j)}\,,\\
 X_{U^{(j)}}^{(j+1)} &  i > U^{(j)},
\end{cases}
\end{align}
where
$L^{(0)}= 0$, $U^{(j)} = L^{(j)} + K$ and 
\begin{align}
L^{(j+1)}&= \sup\{i \le U^{(j)} \colon f(X^{(j)}_{\ell},W_{\ell}) =  f(X^{(j+1)}_{\ell},W_{\ell}) \text{ for }  0\leq \ell< i\}\nonumber\\
&= \sup\{i \le U^{(j)} \colon f(X_{\ell},W_{\ell}) =  f(X^{(j)}_{\ell},W_{\ell}) \text{ for }  0\leq \ell< i\}\nonumber\\
&= \sup\{i \le U^{(j)} \colon X_\ell =  X^{(j)}_\ell \text{ for }  0 \le \ell\le i\},
\qquad\hbox{for }j\geq 0\,,
\label{eq: online picard 2}
\end{align}
where the equivalence of the formulations above follows from \eqref{eq: picard recursion 0}.

The Online Picard algorithm applies the Picard map to the indices $i\in\{L^{(j)}, L^{(j)}+1,\dots,U^{(j)}\}$, where $L^{(j)}$ is an increasing sequence which monitors the position of the last index $i$ such that $X^{(j)}_{0:i}$ has reached its fixed point. 
The algorithm is then stopped as soon as $L^{(j)} \ge N$, see also the pseudocode in Algorithm~\ref{alg: OPA}. While \eqref{eq: online picard 1} is formally defined as a map on $\cX^{\infty}$, the algorithm in practice only operates on $K$ coordinates of the Markov chain 
as detailed in Appendix~\ref{app: pseudo-code} in the supplementary material. The Online Picard algorithm is always more efficient than the classical Picard algorithm (Algorithm~\ref{alg: PA}), since it avoids the vacuous updates of $X^{(j)}_{i}$ for $i\leq L^{(j)}$, see Figure~\ref{fig: illustration picard vs online picard} for an illustration.

\begin{algorithm}[!h]
\caption{Picard algorithm.}  \label{alg: PA}
\KwIn{$N, K \in \NN$, $X^{(0)}_{0:K} \in \cX^{K+1}$.}
 Initialize $j = 0$, $\ell = 0$ and $W_0, W_1,\dots, W_{N-1} \overset{\text{i.i.d.}}{\sim}\nu$\;
 \While{$\ell K < N$}{
 $X^{(j+1)}_{\ell K : (\ell + 1) K} = \Phi( X^{(j)}_{\ell K : (\ell + 1) K}, W_{\ell K : [(\ell + 1) K -1]})$;  \hfill{This step uses parallel processors}\;
$X_i^{(j+1)} = X^{(j)}_{i}$, for $i < \ell K$\;
 \If{$X^{(j+1)}_{\ell K : (\ell + 1) K} = X^{(j)}_{\ell K : (\ell + 1) K}$}{
 $\ell = \ell + 1$\;
 }
$j = j +1$\;
 }
\KwOut{$X^{(j)}_{0:N}$.}
\end{algorithm}

\begin{algorithm}[!h]
\caption{Online Picard algorithm}  \label{alg: OPA}
\KwIn{$N, K \in \NN$, $X^{(0)}_{0:K} \in \cX^{K+1}$}
 Initialize $j = 0$, $L^{(0)} = 0$ and $W_0, W_1,\dots, W_{N-1} \overset{\text{i.i.d.}}{\sim}\nu$\;
\While{$L^{(j)} < N$}{
$U^{(j)} = L^{(j)} + K$\;
$X^{(j+1)}_{L^{(j)}:U^{(j)}} = \Phi( X^{(j)}_{L^{(j)}:U^{(j)}}, W_{L^{(j)}:U^{(j)}-1})$; \hfill{This step uses parallel processors}\;
$L^{(j+1)}= \sup\{i \le U^{(j)} \colon f(X^{(j+1)}_{\ell},W_{\ell}) =  f(X^{(j)}_{\ell},W_{\ell}) \text{ for }  0\leq \ell< i\}$\;
$X_i^{(j+1)} = X^{(j+1)}_{U^{(j)}}$, for $i > U^{(j)}$, $X_i^{(j+1)} = X^{(j)}_{i}$, for $i \le L^{(j)}$\;
$j = j + 1$\;
}
\KwOut{$X^{(j)}_{0:N}$.}
\end{algorithm}

\begin{figure}
    \centering
    \includegraphics[width=0.4\linewidth]{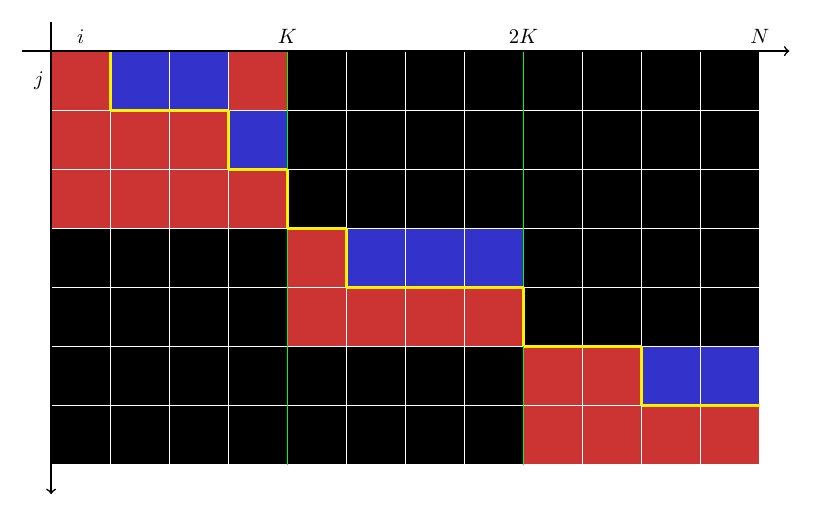}
     \includegraphics[width=0.4\linewidth]{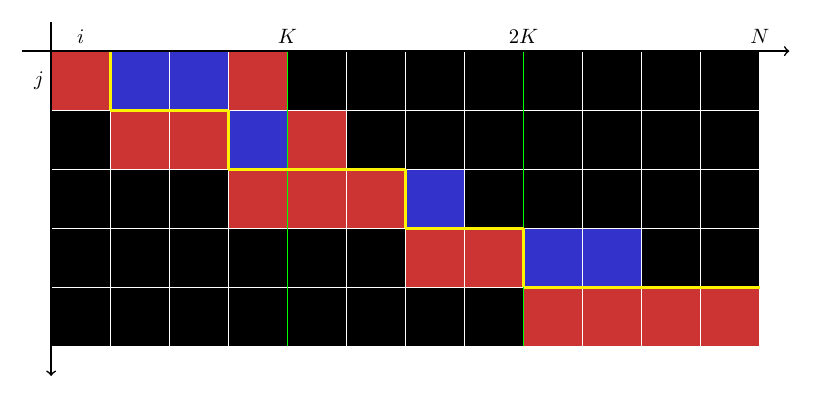}
    \caption{Illustration of the classical Picard algorithm (Algorithm~\ref{alg: PA}) applied sequentially  to every block of length $K$ (left grid) vs. the Online Picard algorithm (Algorithm~\ref{alg: OPA}, right grid). The color of the $(j,i)$ entry of each grid represents the state of the $i$th step of the Markov chain at the $j$th Picard recursion: red for $f(X^{(j)}_i, W_i) = f(X^{(j-1)}_i, W_i)$ (correct guess), blue for $f(X^{(j)}_i, W_i) \ne f(X^{(j-1)}_i, W_i)$. Black for the increments for which  no processor has been allocated for computing the function $f$. Here, $K = 4$, $N = 3K$. Yellow boundary line in correspondence of $L^{(j)} = \sup\{i \le U^{(j)}\colon X_\ell^{(j)} = X_\ell^{(j-1)}; \ell \le i\}$.}
    \label{fig: illustration picard vs online picard}
\end{figure}

In the next two sections, we study the distribution of the random variables $(L^{(j)})_{j \in \NN}$ for the Online Picard algorithm applied to popular zeroth-order Markov chains under log-concavity, which allows to deduce complexity statements about the algorithm in that setting.

\section{Complexity results under log-concavity}\label{sec: Complexity of Picard algorithms}

\subsection{Assumption on the target}\label{subsec: smoothness}
We study convergence of the Picard map $\Phi$ under the following smoothness assumption on $\pi$. Note that the analysis of the convergence of $\Phi$ does not require strong convexity of $V$. 

\begin{assumption}[Smooth and Hessian Lipschitz target]
\label{ass: 1} 
The target distribution $\pi$ has a Lebesgue density in $\cX=\RR^d$ proportional to $\exp(-V)$ where $V$ is a twice-differentiable, $L$-smooth and $(\gamma L^{3/2}d^{1/2})$-strongly Hessian-Lipschitz for some $\gamma\ge 0$, meaning that the Hessian of $V$ satisfies
\begin{align*}
 -L \mathrm{I}_d &\preccurlyeq \nabla^2 V(x) \preccurlyeq  L \mathrm{I}_d\, , &x \in \cX    \\
 -M \|y-x\| &\preccurlyeq \nabla^2 V(y) - \nabla^2 V(x) \preccurlyeq  M \|y-x\|\,, &x,\, y \in \cX
\end{align*}
for some matrix $M$ with $\mathrm{Tr}(M) \le \gamma L^{3/2}d^{1/2}$, where $A \preccurlyeq B$ if $A-B$ is positive semi-definite for two symmetric matrices $A,B$.
\end{assumption}
The strongly Hessian-Lipschitz condition implies the more classical Hessian-Lipschitz condition, $\|\nabla^2 V(y) - \nabla^2 V(x)\| \le m^\star \|y-x\|$, where $m^\star = \lambda_{max}(M)$, while it is weaker than assuming a bounded Frobenius norm of the third derivative, as previously used for analysis of MCMC algorithms \citep[cf. ][]{chen2023does}. This assumption holds, for example, for standard Bayesian logistic regression models with Gaussian priors, see Section~\ref{sec:logistic_regression} in the supplement for details. In general, we acknowledge that Assumption~\ref{ass: 1} can be restrictive 
and, while needed for the theoretical analysis, we do not expect it to be necessary for the proposed Picard algorithm to perform well in practice. 
For example, the numerical simulations in Section~\ref{sec: numerical illustrations} suggest that our theoretical results can be predictive of practical performances beyond the assumptions required to derive them (see e.g.\ the examples on Poisson regression models therein).

\subsection{Metropolis-Hastings kernels}\label{subsec: mh kernels}
In this work, we focus on two popular zeroth-order Metropolis-Hasting Markov chains: \emph{Random Walk Metropolis (RWM)} \citep{metropolis1953equation} and \emph{Metropolis within Gibbs (MwG)} \citep{geman1984stochastic}. 
Both MCMC have piecewise constant Picard maps and can be simulated with the Online Picard algorithm.
We start with the RWM case and defer the MwG case to Section~\ref{sec: Metropolis within gibbs}.

\begin{assumption}[RWM]
\label{ass: 1.2} 
$X_{0}, X_1,\dots$ is a RWM Markov chain on  $\cX = \RR^d$ with stepsize $\xi = h/(\sqrt{L d})$ for some $h > 0$, i.e.\ it follows \eqref{eq: 1} with $W=(Z,U)$, $U \sim \mathrm{Unif}([0,1])$ uniformly distributed, $Z \sim \cN(0,  h^2I_d/(Ld))$ independently of $U$,  
\begin{align}\label{eq:zero_f}
f(x, W)&= Z B(x, U, Z)
\quad\hbox{with }B(x, u, z) = \ind\left(\pi(x + z)/\pi(x) \ge u\right)
\end{align}
and $\ind(\cdot)$ being the indicator function. 
\end{assumption}

By \eqref{eq: online picard 2} and \eqref{eq:zero_f} combined we obtain
\begin{align*}
L^{(j+1)} &= \sup\{i \le U^{(j)} \colon B(X_\ell, U_\ell, Z_\ell)
=  
B(X^{(j)}_\ell, U_\ell, Z_\ell) \text{ for }  0 \le \ell< i\}\,, &j \ge 0\,,
\end{align*}
which means that, in the RWM case, the Picard convergence is entirely dependent on the probability of `correctly guessing' the $\{0,1\}$-valued random variables
$B(X_i, U_i, Z_i)$.

\subsection{Main results}\label{sec: main results}
Let $(X_i)_{i=0,1,\dots}$ and $(X^{(j)}_i)_{i,j=0,1,2,\dots}$ be as in \eqref{eq: 1} and \eqref{eq: online picard 1}, respectively.

\begin{theorem}[Probability of an incorrect guess]\label{thm: main theorem 1}
Under  Assumptions~\ref{ass: 1}-\ref{ass: 1.2}, for all $x_0 \in \cX$, $w_0 \in \cW$, we have
\begin{align}
    \label{eq: bounding probability of incorrect guess}
     \PP(f(X^{(j)}_i,  W_i) \ne f(X_i, W_i)\mid X_0 = x_0, W_0 = w_0) &\le c_0\frac{i}{d}+\delta(d)+2^{-j}
\end{align}
for all $0 \le i \le \min(d,K)$ and $j\in\{0,1,\dots\}$, with $c_0 = 15 h^4 (\sqrt{\frac{2}{\pi}} + \frac{h \gamma}{2})^2 $,
$\delta(d) = \frac{5}{3}\exp(-3d/2)$. The probability in \eqref{eq: bounding probability of incorrect guess} is with respect to the randomness of $(W_1,W_2,\dots)$. 

\end{theorem}
Theorem~\ref{thm: main theorem 1} essentially states that, after a logarithmic number of Picard recursions, the probability that the Picard map incorrectly predicts the $i$th step is $\cO(i/d)$. 
\begin{remark}
By Theorem~\ref{thm: main theorem 1}, the probability of correctly guessing individual increments can be controlled up to $i=\cO(d)$.
Assuming \eqref{eq: bounding probability of incorrect guess} to be tight, this also suggests that increasing the number of parallel processors beyond $K = \cO(d)$ provides little gains in terms of speeding up convergence of the Picard recursion, as numerically observed in Section~\ref{sec: numerical illustrations}.
\end{remark}

\begin{remark}
The upper bound of Theorem~\ref{thm: main theorem 1} becomes vacuous for $i > d$.
This is consistent with the fact that the mixing time of RWM under Assumptions~\ref{ass: 1} is $\cO(d)$. Indeed, if $X_0$ and $X_i$ are approximately independent, then the prediction of $B(X_i, U_i, Z_i)$ with $B(X_0, U_i, Z_i)$ is uninformative. 
\end{remark}

By the union bound, Theorem~\ref{thm: main theorem 1}  implies that, with high probability, the Online Picard algorithm makes $\cO(\sqrt{d})$ consecutive correct guesses after $\cO(\log(d))$ calls of the Picard map when using $K = \cO(\sqrt{d})$ processors, as shown below.

\begin{corollary}
\label{corol: picard map}
Under  Assumptions~\ref{ass: 1}-\ref{ass: 1.2}, for all $x_0\in \cX, \, w_0\in \cW$, $\epsilon > 0$, $K \le \sqrt{\epsilon^2 d/(2 c_0)}$, $c_0$ as in Theorem~\ref{thm: main theorem 1}, $j = \lceil \log(d)\rceil$, $h \ge 1$,  we have that
\begin{align*}
    \PP (L^{(j)} \ge K \mid X_0 = x_0, W_0 = w_0) \ge 1 - \epsilon\, . 
\end{align*}    
\end{corollary}

Notice the validity of the bound in Corollary~\ref{corol: picard map} for any starting value $X_0$ and  innovation $W_0$.
This allows to prove the next main theorem, where we bound with high probability the number of parallel iterations required by the Online Picard algorithm to simulate a given number of steps of RWM with $K$ up to $\cO(\sqrt{d})$ parallel processors.
\begin{theorem}
\label{thm: complexity rwm} 
Suppose Assumptions~\ref{ass: 1}-\ref{ass: 1.2} hold with $h \ge 1$. Then,  for all $\epsilon > 0$, $N\in\NN$, $K \le \sqrt{d/(8 c_0)}$, $c_0$ as in Theorem~\ref{thm: main theorem 1}, we have that $\PP(L^{(j)} > N) \ge 1 - \epsilon$ with 
$$j = \left\lceil \log(d)
    \max\{4N/K,2\log(1/\epsilon)\}
    \right\rceil
    =\mathcal{O}\left(\frac{N}{K}\right)
    \,.$$
\end{theorem}
By Theorem~\ref{thm: complexity rwm}, the Online Picard algorithm achieves the optimal speedup relative to the sequential algorithm for $K$ up to $\cO(\sqrt{d})$, since it simulates $N$ steps of the sequential algorithm with $\cO(N/K)$ parallel iterations and $K$ parallel processors. Notice that the convergence of the Picard map does not deteriorate with the condition number $L/m$ of the target.

Combining Theorem~\ref{thm: complexity rwm} with available explicit bounds on the mixing times of RWM \citep{andrieu2024explicit} gives the following complexity result. 
\begin{corollary}
[Parallel round complexity of Online Picard algorithm]
    \label{corollary: complexity rwm} 
    Suppose Assumptions~\ref{ass: 1}-\ref{ass: 1.2} hold
    with $m I_d \preccurlyeq \nabla^2 V$, $h = 1$ and 
    $X_0 \sim \cN(x^\star, L^{-1}I_d)$, where $x^\star$ is the minimizer of $V$. 
    Then, for any $\epsilon > 0$, the Online Picard algorithm with $K \le \lfloor \sqrt{d/(8c_0)}\rfloor$, $c_0$ as in Theorem~\ref{thm: main theorem 1}, outputs a random variable $X$, 
    with $\|\cL(X)-\pi\|_{\mathrm{TV}} \le \epsilon$ after 
    $$J= \cO\left(\frac{L}{m}\frac{d}{K}\,\mathrm{polylog}(\epsilon^{-1})\right)
    $$
    parallel iterations, where $\cL(X)$ denotes the law of $X$ and $\|P - Q\|_{\mathrm{TV}} = \sup_{A}|P(A) - Q(A)|$ for two probability distributions $P$ and $Q$.
\end{corollary}
Corollary \ref{corollary: complexity rwm} implies that the Online Picard algorithm with $K = \cO(\sqrt{d})$ achieves a speedup of $\cO(\sqrt{d})$ compared to the sequential algorithm.

The convergence bounds in Theorem~\ref{thm: complexity rwm} apply to any initialization $X_0$. With the next proposition, we show that the convergence is faster when the chain is in the tails of the distribution. 
\begin{proposition}     
\label{prop: transient phase} 
Suppose Assumptions~\ref{ass: 1}-\ref{ass: 1.2} hold, $m I_d \preccurlyeq \nabla^2 V$ and $X^{(0)}_\ell= x_0$ for all $\ell\geq 0$. 
Then, for all $j \ge 0$ and $K > 0$, we have
\begin{align*}
\lim_{\|x_0\|\to\infty}\PP\left(L^{(j)}=jK\right) = 1.
    \end{align*}
\end{proposition}
Proposition \ref{prop: transient phase} implies that, as the initial point goes further into the tails, the convergence of the Picard map becomes faster (up to converging in one iteration). This is partly related to classical scaling limits obtained in the transient phase for RWM \citep{christensen2005scaling}, which show that RWM in the tails has a deterministic behavior. In our context, Proposition \ref{prop: transient phase} suggests that the Picard map converges faster in the transient phase, which we do observe in the numerical simulations of Section~\ref{sec: tails experiment}.

\subsection{Contraction of the Picard map: comparison between first- and zeroth-order case}
\label{sec: technical comparison picard maps}
In this section we compare the contraction of the Picard maps applied to RWM to the one obtained in the Unadjusted Langevin algorithm (ULA) case \citep{anari2024fast}. 
The comparison illustrates how the first-order unadjusted case, where the map $x \mapsto f(x, W)$ is smooth, differs from the zeroth-order case of RWM, where $x \mapsto f(x, W)$ is piecewise constant.

Denote by  $\Phi^{\mathrm{ULA}} \colon \cX^{K+1} \times \cW^{K} \mapsto \cX^{K+1}$ the Picard map applied to the ULA Markov chain, defined as in \eqref{eq: 1} 
    with $$
f^{\mathrm{ULA}}(X, W) = -\frac{\xi^2}{2}\nabla V(X) + \xi W, \quad  W \sim \cN(0, I_d)
$$ 
for a tuning parameter $\xi>0$. 
Under Assumption~\ref{ass: 1}, for all $x, y \in \cX^{K+1},$ with $x_0 = y_0$,  all $w \in \cW^{K}$, and for all $i \le K$, we have
    \begin{align*}
          \|\Phi_i^{\mathrm{ULA}}(x, w)- \Phi_i^{\mathrm{ULA}}(y, w) \| &= \frac{\xi^2}{2}\|\sum_{\ell = 1}^{i-1}( \nabla V(x_\ell) - \nabla V(y_\ell))\|
          \le \frac{\xi^2 L}{2}\sum_{\ell = 1}^{i-1} \| x_\ell - y_\ell \|,
    \end{align*}
    while Lemmas~\ref{lemma: 1}-\ref{lemma: 2} in the supplementary material imply that, under Assumptions~\ref{ass: 1}-\ref{ass: 1.2}, 
    \begin{align}
    \EE[\|\Phi^{\mathrm{RWM}}_i(x,  W) - \Phi^{\mathrm{RWM}}_i( y, W)\|] 
    &\le \sqrt{\EE[\|\Phi^{\mathrm{RWM}}_i(x,  W) - \Phi^{\mathrm{RWM}}_i( y, W)\|^2]} \nonumber \\
    &\le  \sqrt{15 \xi^2 d\sum_{\ell = 1}^{i-1} \PP(f(x_\ell, W_\ell) \ne f(  y_\ell,  W_\ell))+ \delta(d)}\nonumber\\
    &\le  \sqrt{ 15 \xi^3 L d c_1 \sum_{\ell = 1}^{i-1} \|x_\ell - y_\ell\| + \frac{1}{\xi L c_1} \delta(d)} \label{eq: contraction rwm}
    \end{align}
    where $W \sim \nu^{\otimes K}$, with marginals as in \eqref{eq:zero_f}, $c_1 = (\sqrt{\frac{2}{\pi}} + \frac{\gamma \xi \sqrt{L d}}{2})$ and $\delta(d)$ as in Theorem~\ref{thm: main theorem 1}.
Note that the contraction of the Picard map occurs pointwise (for each innovation $w$) for ULA and in expectation for RWM. More crucially, the RWM contraction has a slower rate, mainly due to the square root function in \eqref{eq: contraction rwm}.

\section{Extension to Metropolis within Gibbs}\label{sec: Metropolis within gibbs}
\subsection{Metropolis within Gibbs kernels}

We now extend our convergence results to Metropolis within Gibbs (MwG) Markov chains. 
Our main motivation for considering MwG is that it displays better empirical performance in our numerical simulations (see Appendix \ref{sup: simulations mwg} in the supplementary material).

For a given orthonormal basis $\{o_0,o_1,\dots,o_{d-1}\}$ spanning $\RR^d$ (e.g. the standard basis), the (deterministic-scan) \emph{Metropolis within Gibbs (MwG)} chain on $\RR^d$ \citep{geman1984stochastic} performs iteratively a one-dimensional Metropolis-Hastings update in the subspace spanned by $o_i$, for $i = 0,\dots,d-1$. 
The corresponding chain $(X_i)_{i=0,1,\dots}$ can be treated as a time-inhomogeneous Markov chain, where every block of $d$ coordinates of the Markov chain is iteratively defined as 
\begin{align*}
X_{i+1} 
&= X_{i} + f_i(X_{i}, W_{i}), 
&\text{ for }i=0,1,\dots,d-1,
\end{align*}
with innovations $W_i = (U_i, \,Z_i)$, where $Z_0,\dots,Z_{d-1} \stackrel{\text{i.i.d.}}\sim \nu$ are one-dimensional symmetric random variables,
$U_0,\dots,U_{d-1} \stackrel{\text{i.i.d.}}\sim \mathrm{Unif}([0,1])$, independently of $Z_0,\dots, Z_{d-1}$ and
\begin{align}\label{eq: MwG}
    f_i(x, (U, Z)) &= B(x, U, o_i Z)o_i Z\, , &\text{ for }i=0,1,\dots,d-1,
\end{align}
with $B$ as in \eqref{eq:zero_f}.

\subsection{Results for Metropolis within Gibbs}
If we denote by $\mathrm{Haar}(d)$ the Haar measure on the orthogonal group in  dimension $d$, our theoretical results hold with the following assumption.
\begin{assumption}
    \label{ass: MwG}
    $X_{0},X_1,\dots$ is a MwG chain on $\cX = \RR^d$ with $(o_0,o_1,\dots,o_{d-1}) \sim \mathrm{Haar}(d)$ and $Z = h/\sqrt{L d} (2P - 1)S$, $P \sim \mathrm{Ber}(1/2)$ and $S \sim \chi(d)$, where $\mathrm{Ber}(p)$ and $\chi(d)$ are respectively the Bernoulli distribution with parameter $p$ and the Chi distribution with $d$ degrees of freedom.
\end{assumption}

\begin{theorem}
    \label{thm: gibbs}
    Under Assumptions \ref{ass: 1}  and \ref{ass: MwG},
    \begin{itemize}
        \item  the statement of Theorem~\ref{thm: main theorem 1} holds with $c_0$ and $\delta(d)$ replaced by $c_0 = 2 h^4 (\sqrt{\frac{2}{\pi}} + \frac{h \gamma}{2})^2 $, $\delta(d) = 11\exp(-d/10)$;
    \item the statement of Theorem~\ref{thm: complexity rwm} holds with $h > 5$, $c_0$ as above;
    \item the statement of Proposition~\ref{prop: approximate picard} below holds with $c_0$ as above and $d \ge -10\log(\epsilon r/33)$,  where $r\ge0$ is the tolerance error of the Approximate Picard defined in Section~\ref{sec: approximate online-picard}.
    \end{itemize}
\end{theorem}
Theorem~\ref{thm: gibbs} shows that the analysis of the Picard map derived in Section~\ref{sec: main results} can be extended to MwG. However, for MwG, it is not possible to derive the parallel round complexity (as in Corollary~\ref{corollary: complexity rwm}) since,  to the best of our knowledge, tight complexity bounds for MwG under a feasible start remain unknown, with the only available result known to us being a bound on the spectral gap  provided in  \citet[Corollary 7.4]{ascolani2024entropy}.

The next proposition shows that the Online Picard algorithm, when applied to MwG for isotropic Gaussian targets, achieves optimal speedup equal to $K$ with $K$ parallel processors, since each application of the Picard map converges instantaneously to its fixed point. This might be related to the greater speedups of MwG relative to RWM for well conditioned targets empirically observed in Section~\ref{sec: numerical illustrations}.
\begin{proposition}[MwG with isotropic Gaussian targets] \label{prop: instant convergence ORWM isotropic Gaussians} 
    Suppose Assumption~\ref{ass: MwG} holds  and  $V(x) = \|x - \mu\|^2/(2\sigma^2)$, $\mu \in \RR^d, \sigma^2 > 0$. For every $K \le d$ and $X_0 \in   \cX$, we have $L^{(j)} = jK$.
\end{proposition}
Note that the result in Proposition~\ref{prop: instant convergence ORWM isotropic Gaussians} is deterministic, i.e. it holds for every sequence $w \in \cW^{\infty}$.

\section{Approximate Picard} \label{sec: approximate online-picard}
By Theorem~\ref{thm: main theorem 1} the probability of correctly guessing each individual increment for the Picard map remains small up to $i = \cO(d)$ steps. However, when accumulating all these small probabilities across $i$, the first incorrect guess happens much earlier, at $i = \cO(\sqrt{d})$, see Corollary~\ref{corol: picard map}. 
This creates a bottleneck and does not allow us to leverage effectively more than $K = \cO(\sqrt{d})$ parallel processors, see Theorem~\ref{thm: complexity rwm}. 

This motivates us to propose an Approximate Online Picard algorithm, which can leverage effectively up to $K = \cO(d)$ processors, but whose output does not match anymore the output of the recursion in \eqref{eq: 1} and for which the limiting distribution of the resulting Markov chain has a small bias relative to $\pi$.

By \eqref{eq: online picard 2}, each iteration of the Online Picard algorithm advances up to the position of the first error of the Picard recursion. Here we relax this condition by tolerating a fixed (and small) percentage $r \in [0,1]$ of mistakes.
We thus define the \emph{Approximate Online Picard algorithm} exactly as in \eqref{eq: online picard 1} and \eqref{eq: online picard 2}, but replacing $(L^{(j)},U^{(j)})_{j=0}^\infty$ with $(L_r^{(j)},U_r^{(j)})_{j=0}^\infty$ defined as $L^{(0)}_r=0$, $U_r^{(j)} = L_r^{(j)} + K$ and 
\begin{align*}
L_r^{(j+1)} &= \sup\{L_r^{(j)} < i \le U_r^{(j)} \colon  \cA_\ell^{(j)}  \le r \text{ for all }  L_r^{(j)}< \ell \le i\}
&j\geq 0\,,
\end{align*}
where 
$$
\cA_\ell^{(j)} = \frac{|\{s\in\{L_r^{(j)},L_r^{(j)}+1,\dots,\ell-1\} \colon f(X^{(j+1)}_{s},  W_{s}) \neq f(X^{(j)}_{s},  W_{s})\}|}{\ell-L_r^{(j)}} 
$$ 
is the proportion of mistakes made by the $j$-th Picard recursion in $\{L_r^{(j)}+1,\dots,\ell\}$, see also the pseudocode in Appendix~\ref{app: pseudo-code}.
Here $r$ acts as a tolerance parameter which controls the approximation error, and the case $r=0$ coincides with the (exact) Online Picard algorithm. 
By construction, $L_r^{(j)}$ is non-decreasing in\ $r$, and thus increasing $r$ reduces the runtime. 
Interestingly, for any  $r > 0$, the Approximate Online Picard algorithm is able to exploit up to $K = \cO(d)$ parallel processors effectively, as shown in the following Proposition. 
\begin{proposition}\label{prop: approximate picard}
Under Assumptions~\ref{ass: 1}-\ref{ass: 1.2}, for all $\epsilon>0$, $r \in (0,1]$, $d \ge \log(5/(\epsilon r))$, $j \ge 1.5\log(3/(\epsilon r))$ and $K \le  \epsilon r d /(3c_0)$,  we have
$$
\PP\left(\frac{1}{K}\sum_{i=0}^{K-1}\ind \left(f(X_i^{(j)}, W_i) \ne f(X_i, W_i)\right) \le r\right) > 1 - \epsilon.
$$      
\end{proposition}
Proposition~\ref{prop: approximate picard} shows that the fraction of incorrect guesses can be controlled with high probability up to $K = \cO(d)$ processors and with only $j = \cO(1)$  parallel iterations.
This suggests that the  Approximate Online Picard algorithm with $r > 0$ and $K = \cO(d)$ reduces the parallel run-time of the Online Picard algorithm by a factor $\cO(\sqrt{d})$, thus converging in $j = \cO(1)$ parallel iterations, as numerically verified in Section~\ref{sec: numerical illustrations}.

However, for $r> 0$ the output no longer matches exactly the output of the sequential algorithm. 
Intuitively, we expect the error introduced to be small for small values of $r$, since the increments are coherent with the ones of the original RWM at least $(1-r)$-fraction of the times, so that the Markov chain at those steps follows the correct transition kernel in \eqref{eq: 1}.
Also, the size of the other increments is bounded by the size of $Z$, which has Gaussian tails. 
The error of the Approximate Picard algorithm on the invariant distribution is numerically assessed in Section~\ref{sec: numerical illustrations} for tolerance levels $r$ ranging from 0\% to 20\%, while a rigorous analysis of the error introduced by taking $r > 0$ is outside the scope of this work.

\section{Numerical simulations}\label{sec: numerical illustrations}
\subsection{Overview}
We assess numerically the performance of the (Approximate) Online Picard algorithm for high-dimensional regression models and the SIR epidemic model. The first simulations allow to test performance systematically and assess 
the tightness of the bounds developed in Section~\ref{sec: Complexity of Picard algorithms}. The second experiment illustrates the breadth of our results for a more sophisticated Bayesian problem where the posterior density is not log-concave and is only piecewise continuous. In this setting, state-of-the-art gradient-based methods such as Langevin and Hamiltonian Monte Carlo \citep{RobertsTweedie::1996,neal2011mcmc} cannot be directly used and zeroth-order (gradient-free) methods become a natural choice. 
In this example, we also examine Discontinuous Hamiltonian Monte Carlo (D-HMC) \citep{nishimura2020discontinuous}, a non-reversible zeroth-order variant of MwG. Although our complexity bounds do not apply to D-HMC, the Online Picard algorithm applies directly because its Picard map is piecewise constant. 
In Section~\ref{sec:blackbox} we implement a parallel version of the online Picard algorithm on a concrete real-word application in precision medicine where zeroth-order parallel methods are appropriate (see Section~\ref{sec: parallel architectures} for pseudocode and implementation details). For this application, the log-target can only be evaluated through a black-box computational routine that is expensive and does not provide gradients. 
Practical aspects such as hardware and software considerations, as well as parallelization overheads are discussed in detail in Section~\ref{sec: practical algorithmic aspects} in the supplement. See also \citet{glatt2024parallel, lee2010utility} for related  discussions.

\subsection{Algorithmic specification}
In all numerical simulations, we set $\xi = c^{\mathrm{(RWM)}}/\sqrt{d}$ for RWM,  
and innovation $\tilde Z = \cN(0, c^{\mathrm{(MwG)}})$ for MwG (with standard basis). The stepsizes $c^{\mathrm{(RWM)}}$ and $c^{\mathrm{(MwG)}}$ are tuned to obtain an average Metropolis acceptance-rejection of approximately  23.4\% and 40\%, respectively, as prescribed by the literature \citep{gelman1997weak, gelman1996efficient}. We set the number of integration steps for D-HMC to $L = 3$ and the total number of iterations to $\lfloor N /(d L)\rfloor$, where $N$ is the total number of steps in RWM and MwG, ensuring that the total number of density evaluations is comparable across algorithms. To guarantee ergodicity, the stepsize of D-HMC is drawn as  $h=|z|$ with $z \sim \cN(0, c^{\mathrm{(MwG)}})$ \citep[see][Section 4.2]{nishimura2020discontinuous}. The code to reproduce the simulations is available at \texttt{https://github.com/SebaGraz/ParallelMH.git}.

\subsection{Performance metric}
We define the empirical speedup relative to sequential implementations as the average number of sequential steps simulated per parallel round, i.e.
\begin{equation}
    \label{eq: hat G}
    \hat G = L^{(J)}/J\,,
\end{equation}
where $J$ is the total number of parallel iterations of the (Approximate) Online Picard algorithm. $\hat G$ can be interpreted as the multiplicative speedup obtained by parallel computing when ignoring overheads due to parallel execution. 
In particular, we expect $\hat G$ to match the theoretical results in Theorems \ref{thm: complexity rwm} and \ref{thm: gibbs} and Proposition~\ref{prop: approximate picard}, that is $\hat G = \cO(\sqrt{d})$ for the Online Picard algorithm with $K = \cO(\sqrt{d})$ and $\hat G = \cO(d)$ for the Approximate Picard algorithm with $K = \cO(d)$. Moreover, we verify numerically that increasing $K$ beyond $\cO(d)$ yields no further gains, as suggested by Theorem~\ref{thm: main theorem 1}. 
The gains obtained by increasing $K$ from $\cO(\sqrt{d})$ to $\cO(d)$ for the Online Picard algorithm are more case-dependent, see details in Section \ref{sec: high-dimensional regressions} below.

Note that \(\hat{G}\) is only a measure of speedup \emph{relative to} the original sequential MCMC method.  To evaluate the overall performance of a parallel MCMC method, the value of $\hat{G}$ should be combined with (i) a measure of statistical efficiency of the original MCMC, such as the Effective Sample Size (ESS), and (ii) the wall-clock time required to run the algorithm. In particular, when target evaluations are inexpensive (as in the toy regressions considered in Section~\ref{sec: high-dimensional regressions}), the overhead induced by parallelization (e.g., latency from synchronization and communication) may dominate, slowing the Picard algorithm. In Section~\ref{sec: sir} we report ESS values for the algorithms considered, and in Section~\ref{sec:blackbox} we report wall-clock times for parallel implementations of  the Picard algorithm 
which also account for parallelization overheads.

\subsection{High-dimensional regressions}\label{sec: high-dimensional regressions}
We consider three high-dimensional regression models:
\begin{itemize}
    \item[\textbf{E1}] \emph{Linear regression:} for $i =1,2,\dots,n$, $
Y_i \mid A_i \overset{\text{i.i.d.}}{\sim} \cN(\langle A_i,x\rangle , \sigma^2), \sigma^2 = 1
$ and unknown parameter $x \in \RR^d$.
The sample size is set to $n = 5d$. 
\item[\textbf{E2}] \emph{Logistic regression:} for $i =1,2,\dots,n$, $Y_i \mid A_i \overset{\text{i.i.d.}}{\sim} \mbox{Ber}(\psi(\langle A_i, x \rangle )),$ with $\psi(u) =(1 + e^{- u})^{-1}$ and unknown parameter $x \in \RR^d$. The sample size is set to $n = 10d$.
\item[\textbf{E3}] \emph{Poisson regression:} for $i =1,2,\dots,n$, $Y_i \mid A_i \overset{\text{i.i.d.}}{\sim}\mbox{Poiss}(\exp(\langle A_i, x \rangle )), $ with unknown parameter $x \in \RR^d$. The sample size is set to $n = 10d$.
\end{itemize}
In all cases, we set independent standard Gaussian priors for the parameters. With this setting, the corresponding posteriors are analytical for the linear regression, while not tractable for the logistic and Poisson regression. Furthermore, for the latter, the gradient of $\log \pi$ decays exponentially fast in the tails, making standard first-order Monte Carlo methods numerically unstable. We generate covariates as $A_i\overset{\text{i.i.d.}}{\sim} \cN_d(0, I_d/d)$, some true parameters $x \sim \cN_d(0, I_d)$ and $Y_{1:n}$ from the corresponding model. We run multiple times our algorithms, each time varying the number of processors $K$ and the parameter dimension $d$.
We initialize the algorithms at the posterior mean for linear regression and at the true parameter value otherwise. Unless specified otherwise, we run algorithms until $L^{(J)}\geq N = 10^4$. Figure~\ref{fig: gauss} and Figure~\ref{fig: logistic} summarize the results for RWM.

The left panels display $\hat G$ as a function of $d$ (on a log-log scale), when $K = d$ and $K = \sqrt{d}$.  
The growth rates match the ones predicted by the theory of Section~\ref{sec: Complexity of Picard algorithms}, i.e. $\cO(\sqrt{d})$ for Online Picard and $\cO(d)$ for the approximate version. Note that increasing 
$K$ from $\sqrt{d}$ to $d$ in the Online Picard algorithm improves $\hat G$ only by a constant factor, indicating that, while there might be a practical benefit in setting $K$ beyond $\sqrt{d}$, the rate in Theorem~\ref{thm: complexity rwm} is tight for RWM.

The right panels display $\hat G$ as a function of $K$ for fixed $d$ ($d = 500$ for linear regression, $d = 200$ otherwise). Again, the results align with the theory, with a linear speedup up to $K = \sqrt{d}$ (blue vertical dashed line) followed by a slower increase until $K = d$ and no increase afterwards.  
Similar experiments are performed for MwG and are  available in the supplementary material, Appendix~\ref{sup: simulations mwg}.
For MwG, we generally notice better parallelization performance compared to RWM. In particular, choosing $K$ larger than $\sqrt{d}$ appears to yield more than a constant improvement in speedup. This may be due to the nearly isotropic targets considered in our regression problems. In such settings, the Online Picard algorithms for MwG can achieve optimal speedup as suggested by Proposition~\ref{prop: instant convergence ORWM isotropic Gaussians}.

\begin{figure}[!ht]
    \centering
    \includegraphics[width=0.8\linewidth] {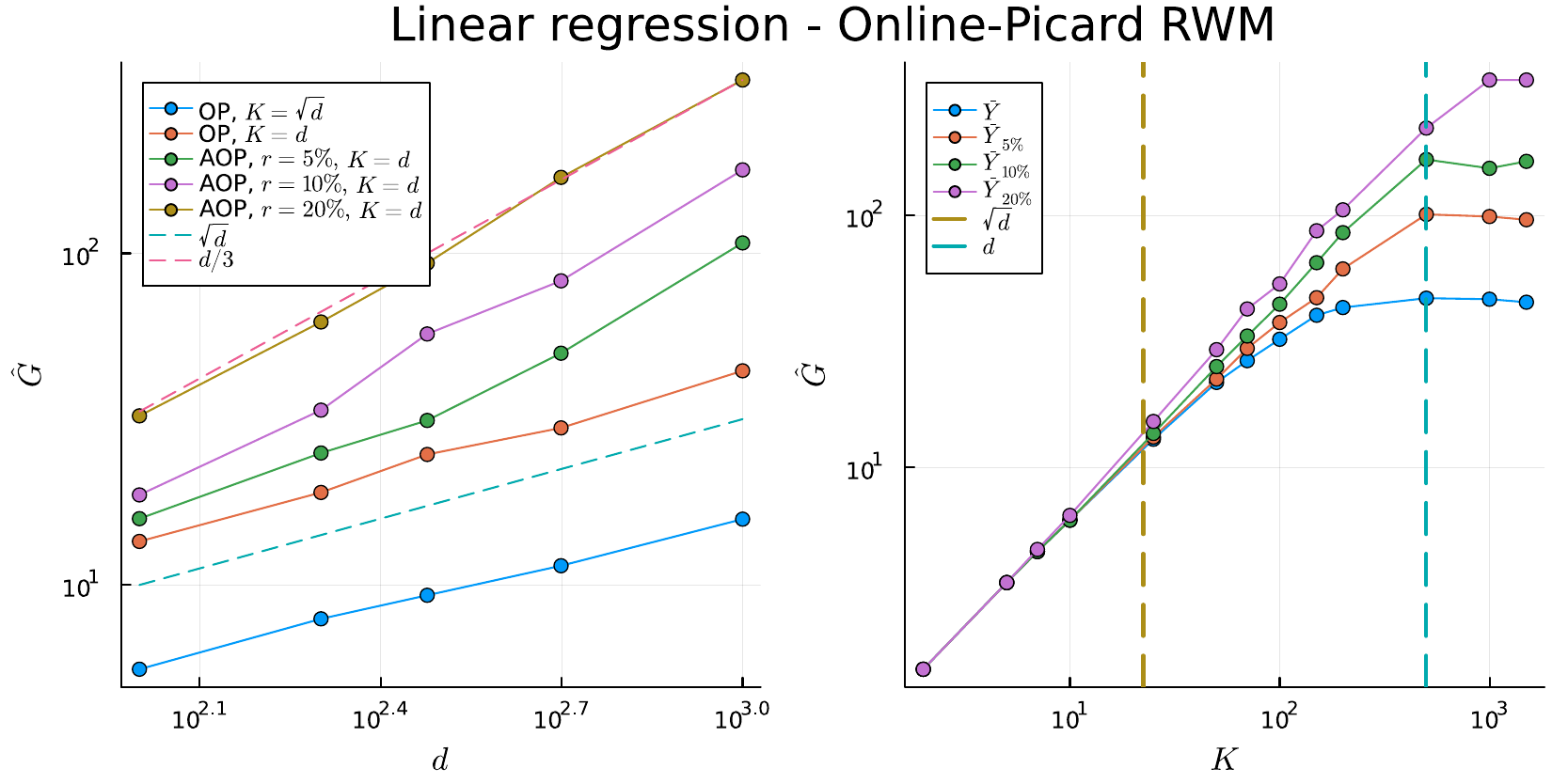}
    \caption{Performance of Online Picard algorithm ($\bar X$) and its approximate versions ($\bar X_r, \, r =5\%,\dots,20\%$) applied to RWM, with target being the linear regression model \textbf{E1}. Left panel: average speedup $\hat G$ ($y$-axis, on a log-scale) with $K \in \{\sqrt{d}, d\}$, $N = 10^4$ and $d = 10^2,\dots,10^3$ ($x$-axis, on a log-scale). Dashed lines $d \mapsto \sqrt{d}$ (blue) and $d \mapsto 3d/20$ (red) for reference. Right panel: average speedup $\hat G$ ($y$-axis, on a log-scale) for $d = 500$, $N = 10^4$ and $K = 2,\dots,1500$ ($x$-axis, on a log scale). Vertical dashed lines for $K = \sqrt{d}$ and $K = d$ for reference.}
    \label{fig: gauss}
\end{figure}

\begin{figure}[!ht]
    \centering
    \includegraphics[width=0.8\linewidth]{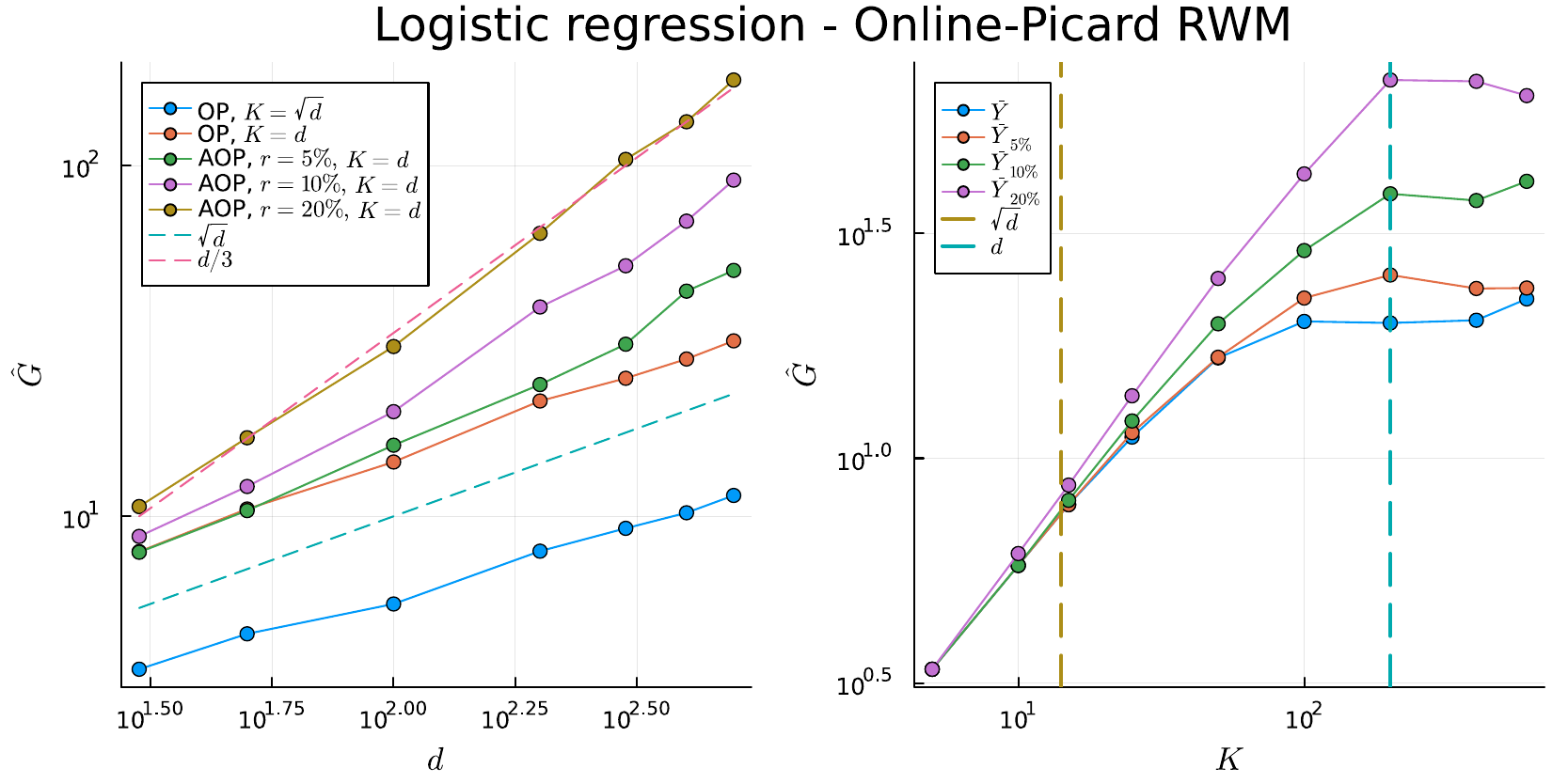} \includegraphics[width=0.8\linewidth]{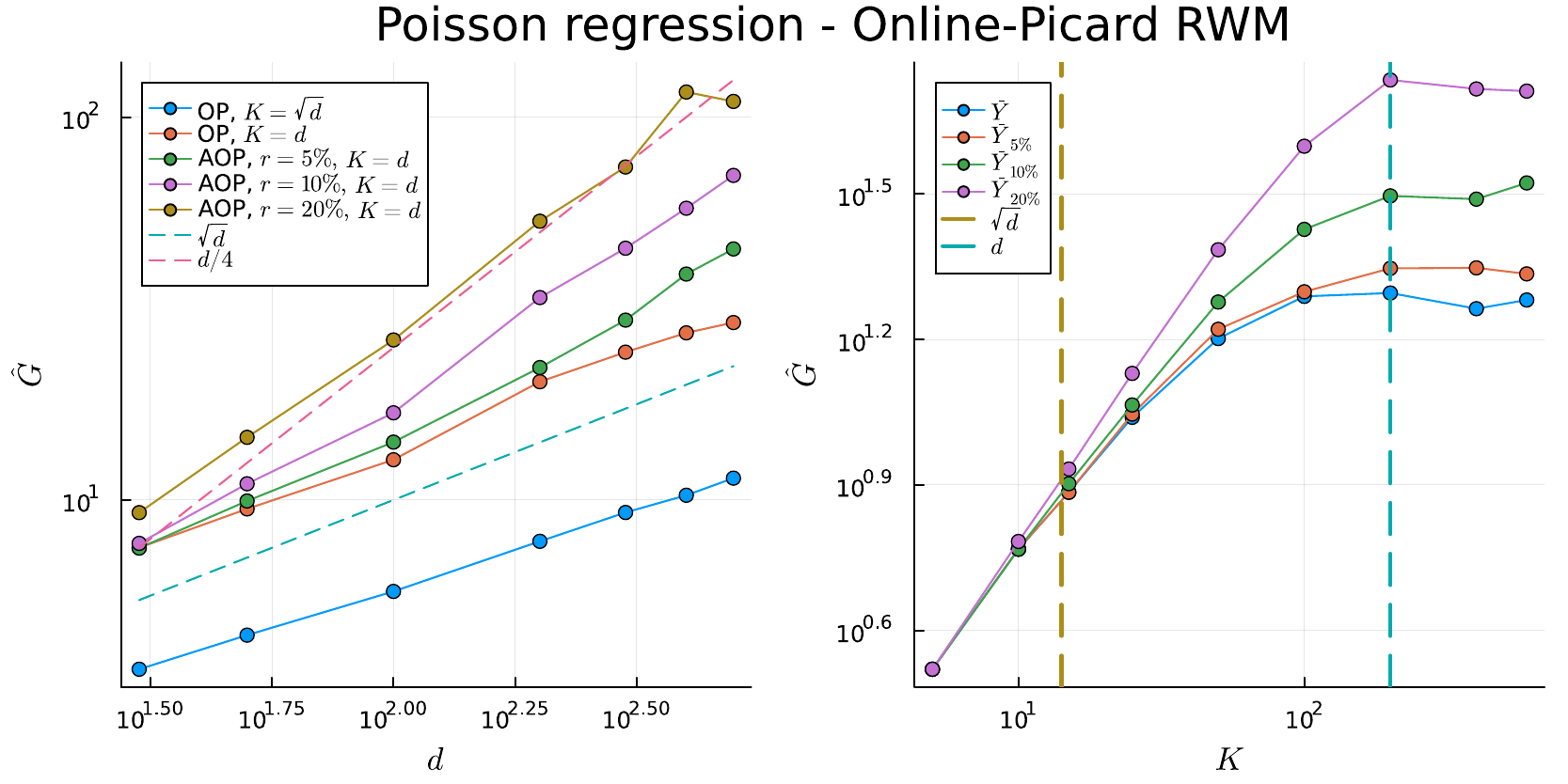}
   
    \caption{Same as Figure~\ref{fig: gauss} for the logistic regression model \textbf{E2}  (top panels) and the Poisson regression model \textbf{E3} (bottom panels), with $N = 10^4, K \in \{\sqrt{d}, d\}$   (left panels). $N = 10^4, d = 200$ (right panels).}
    \label{fig: logistic}
\end{figure}

\subsection{Bias of Approximate Online Picard}\label{sec: bias of approximate online picard}
We evaluate the bias introduced with the Approximate Online Picard algorithms by computing the standardized moments:
\begin{equation}
  \label{eq: error in estimation}  
\cM_r =\sqrt{\frac{1}{d}\sum_{i=1}^d\frac{(\hat \mu^{(r)}_i - \mu_i(\pi))^2}{\sigma_i^2(\pi)}}, \qquad \cE_r = \sqrt{\frac{1}{d}\sum_{i = 1}^d \left(\frac{\hat \sigma_i^{(r)} - \sigma_i(\pi)}{\sigma_i(\pi)}\right)^2}
\end{equation}
where $\hat \mu^{(r)}_i$ is the sample mean and $\hat \sigma_{i}^{(r)}$ is the sample standard deviation of the $i$th element estimated through the Markov chain $X_{0:N}$ simulated with the Approximate Picard algorithm with tolerance level $r$ while $\mu_i(\pi)$ and $\sigma_i(\pi)$ are the mean and standard deviation of the $i$th-marginal posterior distribution (obtained with a very long MCMC run, when not analytically tractable).

Figure~\ref{fig:bias} displays $\cM_r$ and $\cE_r$, defined in \eqref{eq: error in estimation}, as a function of $r$, for $K = d$ and  $d = 100$. When $r = 0$, the Approximate Online Picard algorithm reduces to the standard Online Picard algorithm, and $\cM_0, \cE_0$ correspond to the Monte Carlo estimation error arising from taking a finite number of steps $N$. Interestingly, for RWM, increasing $r$ tends to inflate the variance of the resulting approximate samples, in a way analogous to the bias observed in unadjusted Langevin algorithms \citep{roberts1996exponential, dalalyan2020sampling}, while this is not the case for MwG, which appears to be significantly more robust to larger values of $r$.
For these experiments, we set $N = 10^5$.

\begin{figure}
    \centering
    \includegraphics[width=0.8\linewidth]{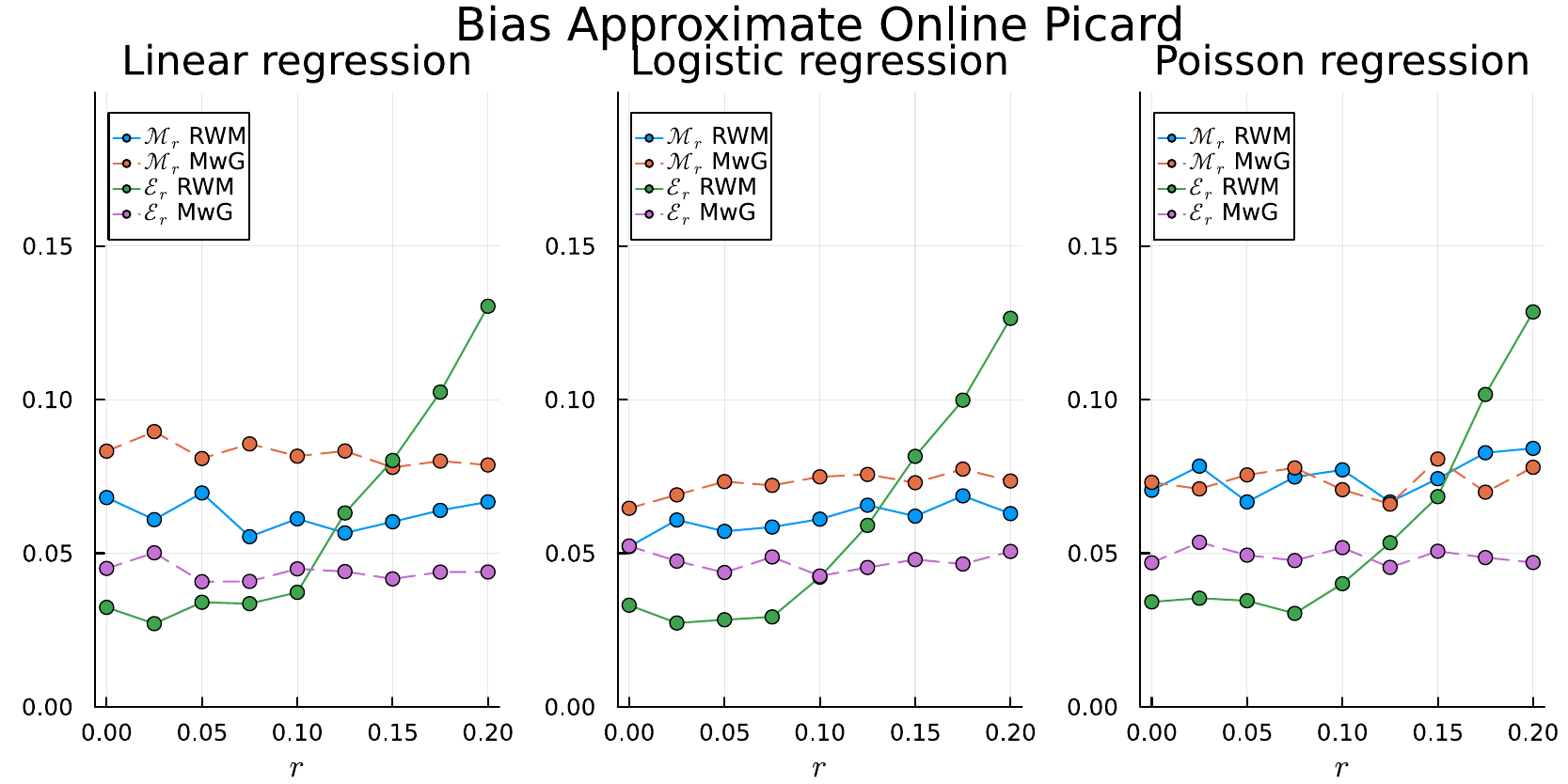}
    \caption{$\cM$ and $\cE$ as in \eqref{eq: error in estimation} for the Approximate Picard algorithms with tolerance ($x$-axis) ranging from $0\%$ (exact Picard algorithm) to $20\%$ for $d = 10^2,\,  K = d, \, N = 10^5$. $N = 4 \times 10^5$ (right panels) for the linear regression \textbf{E1} (left panel), logistic regression \textbf{E2} (middle panel), Poisson regression \textbf{E3} (right panel). Dashed lines for MwG.}
    \label{fig:bias}
\end{figure}

\subsection{Convergence in the tails}\label{sec: tails experiment}

We numerically assess the convergence of the Picard map for RWM and MwG chains started in the tails of the target distribution. Figure~\ref{fig: convergence in the tails} numerically shows the value of $L^{(1)}$, averaged over 10 simulations, for $K = 200$, $d = 200$, random initializations $X_0 \sim \text{Unif}(x \in \RR^d \colon \|x-x^\star\| = s)$, with $x^\star$ being the posterior mode, and target distribution given by posterior distribution of the logistic  model \textbf{E2} in Section~\ref{sec: high-dimensional regressions}. In all experiments, the tuning parameter $h$ was kept fixed and tuned in a preliminary run with the algorithm in stationarity.
Our results suggest that $L^{(1)}$ converges to $K$ as $s$ increases, which is coherent with Proposition~\ref{prop: transient phase}.

\begin{figure}
    \centering
    \includegraphics[width=0.5\linewidth]{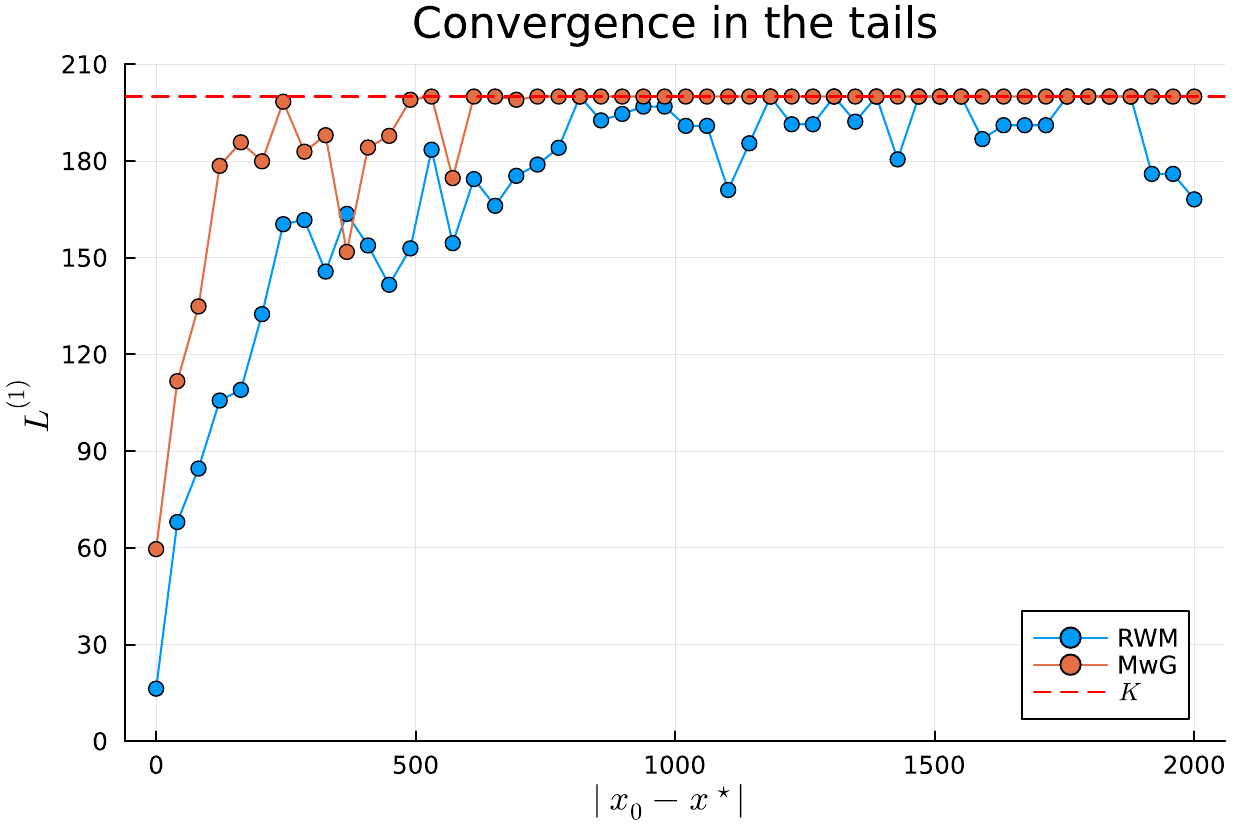}
    \caption{Empirical averages of $L^{(1)}$ ($y$-axis) for $X_0 \sim \text{Unif}(x \in \RR^d \colon \|x- x^\star\| = s)$, $s$ ranging from 0 to 2000 ($x$-axis) and target given by logistic regression \textbf{E2} with $d = K = 200$.}
    \label{fig: convergence in the tails}
\end{figure}

\subsection{SIR epidemic model}\label{sec: sir}

Let  $\{Y(t) \in \{S, I, R\}^M, t\ge 0\}$ be an infection process in a population of size $M$. Each coordinate $Y_i(t)$ takes values
\begin{align*}
y_i(t) &= \begin{cases}
    S & \text{if $i$ is \emph{susceptible} at time $t$}\\
    I & \text{if $i$ is \emph{infected} at time $t$}\\
    R & \text{if $i$ is \emph{removed} at time $t$}\\
\end{cases},& i = 1,2,\dots,M.
\end{align*}
Each coordinate process $t \to Y_i(t)$ is allowed to change state in the following direction:  $S \to I \to R$. In particular the coordinate process changes its state from $S$ to $I$ according to an inhomogeneous Poisson process with rate $t \to \beta_0 \cI_t$ where $\beta_0>0$ and $\cI_t = \sum_{i=1}^M \ind(y_i(t) = I)$ is the total number of infected individuals at time $t$. The transition from $I$ to $R$ of each individual occurs according to an exponential random variable with rate equal to $\gamma_0>0$. The parameters $\beta_0$ and $\gamma_0$ are fundamental as they jointly determine the \emph{reproduction number} $R_0 = M \beta_0/\gamma_0$. See \cite{baily1975mathematical} for a thorough mathematical description of this model and its extensions. 

Let $d\le M$ be the total number of individuals that have been infected during the epidemic. We assume that the epidemic has originated with 1 initial infected individual and has been observed until the epidemic becomes extinct, that is, after the last infected individual has been removed. 

As it is common for this class of models, we do not observe the infection times (i.e. when the coordinate process changes from $S \to I$) and we only observe the times where infected individuals change their state from $I$ to $R$. Hereafter,  we denote by $x \in \RR^{d}$ the \emph{latent} vector of infection times and by $x^\circ \in \RR^{d}$  with $x_i \le x^\circ_i,  \, 1 \le i \le d$, the \emph{observed} removed times of the infected population. We assume that $\gamma_0, \beta_0$ are unknown.   The likelihood for $(\gamma_0, \beta_0)$ becomes tractable only when the infection times of each infected individual are known and is given by
$$
\exp\left(-\beta A(x, x^\circ) + \sum_{\substack{i=1 \\ i\neq k(x)}}^d \log(\beta \cI_{x_i-})\right)\\ 
\left(\prod_{i=1}^d \gamma \exp(-\gamma (x^\circ_i - x_i))  \ind (x_i \le x^{\circ}_i)\right).
$$
Here $k  = \arg \min(x)$ is the label of the first infected individual, $t- = \lim_{s\uparrow t} s $ the left-hand limit and 
\begin{equation}
\label{eq: A(x)}
  A(x, x^\circ) = \int_{x_k}^T \cI_t \cS_t \dd t = \sum_{j=1}^d \left((M-d)  (x_j^\circ - x_j) + \sum_{i=1}^d (x_j^\circ \wedge x_i - x_j\wedge x_i)\right),  
\end{equation}
with $T = \max(x^\circ_1,x^\circ_2,\dots,x^{\circ}_d)$ and $\cS_t = \sum_{i=1}^M \ind(y_i(t) = S)$, is the total number of person-to-person units of infectious pressure exerted during the course of the epidemic, see \cite[Section 5]{neal2005case} for details of the likelihood and posterior measure in the set-up considered here. 

We set two independent Gamma priors for $\gamma$ and $\beta$  with prior parameters  $(\nu_\beta, \lambda_\beta)$ and $(\nu_\gamma, \lambda_\gamma)$ respectively.  We then integrate $\gamma, \beta$ out the posterior and get the marginal posterior density of the latent space of infection time:
\begin{multline}
    \label{eq: posterior SIR}
\pi(x) \propto \left(\prod_{\substack{i=1 \\ i\neq k}}^d \cI_{x_i-}\right)\left(\lambda_\beta + A(x, x^\circ)\right)^{-(d + \nu_\beta - 1)} \\ \left(\lambda_\gamma + \sum_{i=1}^d (x_i^\circ - x_i)\right)^{-(d + \nu_\gamma)} \left(\prod_{i=1}^d \ind (x_i \le x^\circ_i)\right).
\end{multline}
This posterior density is only piecewise continuous, with regions of discontinuity in correspondence to changes of the order of the  combined vector of infection and removal times $(x, x^\circ)$. Once $x$ is simulated, we can obtain posterior samples of $\gamma$ and $\beta$ generating them from the conditional distributions 
\begin{equation}
    \label{eq: conditional posterior}
    \gamma \mid x \sim \Gamma(d + \nu_\gamma, \, 1/(\lambda_\gamma + B(x, x^\circ) )), \qquad \beta \mid x \sim  \Gamma(d + \nu_\beta -1, \, 1/(\lambda_\beta + A(x, x^\circ)))
\end{equation}
where $B(x, x^\circ) = \sum_{i=1}^d (x^\circ_i - x_i)$ and $A(x, x^\circ)$ as in \eqref{eq: A(x)}.

We simulate synthetic data and adopt the same experimental setting of \citet[Section 5.2]{neal2005case}, that is we set
$M = 200$, $\beta = 0.001$ and $\gamma = 0.15$. With this setting, the reproduction number is equal to $R_0 = 4/3$. The forward model generates $d = 98$ infected individuals before the epidemic went extinct.

We set weak uninformative priors for $\beta$ and $\gamma$ with parameters $\lambda_\beta = \lambda_\beta = 0.001$  and $\nu_\beta = \nu_\beta = 1$ and run the Online-Picard applied to RWM, MwG and D-HMC \citep{nishimura2020discontinuous} with target density as in \eqref{eq: posterior SIR} and for $N = 10^6$ steps and with $K = \lfloor \sqrt{d} \rfloor = 9$ parallel processors. The chain is initialized out of stationarity, at $X_{0,i} \sim x^{\circ}_i - \text{Exp}(0.05)$  (verifying that the posterior evaluated at this initial value is non-zero, that is $\pi(X_0) \ne 0$). Moreover, we perform a second and computationally more intensive numerical experiment with $M = 400, \, R_0 = 8/3, \, d = 372,\, N = 5\times 10^6, \, K = \lfloor \sqrt{d} \rfloor = 19$.

Table~\ref{tab: results SIR} shows the results obtained. Here $\hat G$ is computed as in \eqref{eq: hat G},  $\cM_0$ and $\cE_0$ as in \eqref{eq: error in estimation}, and $d$ denotes the dimensionality of $x$. The Effective Sample size, which is defined in Section~\ref{sec: ESS} of the supplementary material, is estimated using the \emph{effectiveSize} function of the \emph{coda} R package \citep{coda}, applied to each coordinate of the Online Picard output (i.e.\ the full trajectory after burn-in). We report both the average and the minimum across coordinates. 
These values are complementary to those of $\hat G$ (as they do not directly relate to the parallelization performance of the Picard algorithm) and provide an empirical basis for comparing the three MCMC methods under consideration. 

In the examples considered here, D-HMC achieves the highest ESS value, followed relatively closely by MwG, while RWM is significantly less efficient for this problem.
The parallelization speedups (measured by $\hat G$) vary between 4 and 10 for all methods, are highest for MwG and increase with $d$ and $K$, in accordance with theory. We also report values of ESS$\times \hat G$, which represent the ESS per parallel iteration and thus combines MCMC efficiency with the parallelization performance of the Online Picard algorithm. Under this metric, D-HMC outperforms MwG, albeit only by a small margin. We set all Markov chains with burn-in equal to $N / 2$, while the true means and variances were estimated with a long run ($N = 10^7$ for the first  numerical experiment and $N = 5 \times 10^7$ for the second one, burn-in equal to $N / 5$ in both cases) of RWM, started at the true values $x$. Figure~\ref{fig: sir} shows the functionals $B(X, x^{\circ})$ and $A(X, x^{\circ})$ as in \eqref{eq: conditional posterior}, evaluated at every 10 steps of the RWM chain, with time index displayed with a yellow-blue gradient color. The right panel, shows the corresponding values of $(\beta,\gamma)$, simulated from the conditional posteriors in  \eqref{eq: conditional posterior}. 

\begin{table}[!t]
\caption{Online Picard applied to SIR model with  $d = 98$ ($K = 9,\, N = 10^6$) and $d = 372$ ($K = 19,\, N = 5\times 10^6$). Under ESS, we report the average ESS and minimum ESS in brackets. Bold characters for the best performance. \label{tab: results SIR}}%
\begin{tabular*}{\columnwidth}{@{\extracolsep\fill}cccccc@{\extracolsep\fill}}
 \hline
Markov chain & $\hat G $ & $\cM_{0}$ & $\cE_0$&  ESS & ESS$\times \hat G$.\\
 \hline
RWM $d = 98$ & 4.28 & 0.22 & 0.25 & {36.6 (4.69)} & 156.65 (20.07)\\
MwG $d = 98$& 4.94 & 0.09 & 0.12 & 404.57 (151.81)& 1998.58 (749.94) \\
 D-HMC $d = 98$&  3.98 &  0.08 &  0.09 &   508.15 (207.60)&  2022.44 (826.25)\\
\hline
RWM $d = 372$ & 5.15 & 0.39 & 0.26 & 18.58 (1.62)& 44.18 (13.90)\\
MwG $d = 372$& 9.70 & 0.12 & 0.10 & 169.68 (22.61)& 1645.89 (219.32)\\
D-HMC $d = 372$&  7.40 &  0.11 &  0.08 &  270.05 (42.13)&  1998.37 (311.76)\\
\hline
\end{tabular*}
\end{table}

\begin{figure}
    \centering    \includegraphics[width=0.8\linewidth]{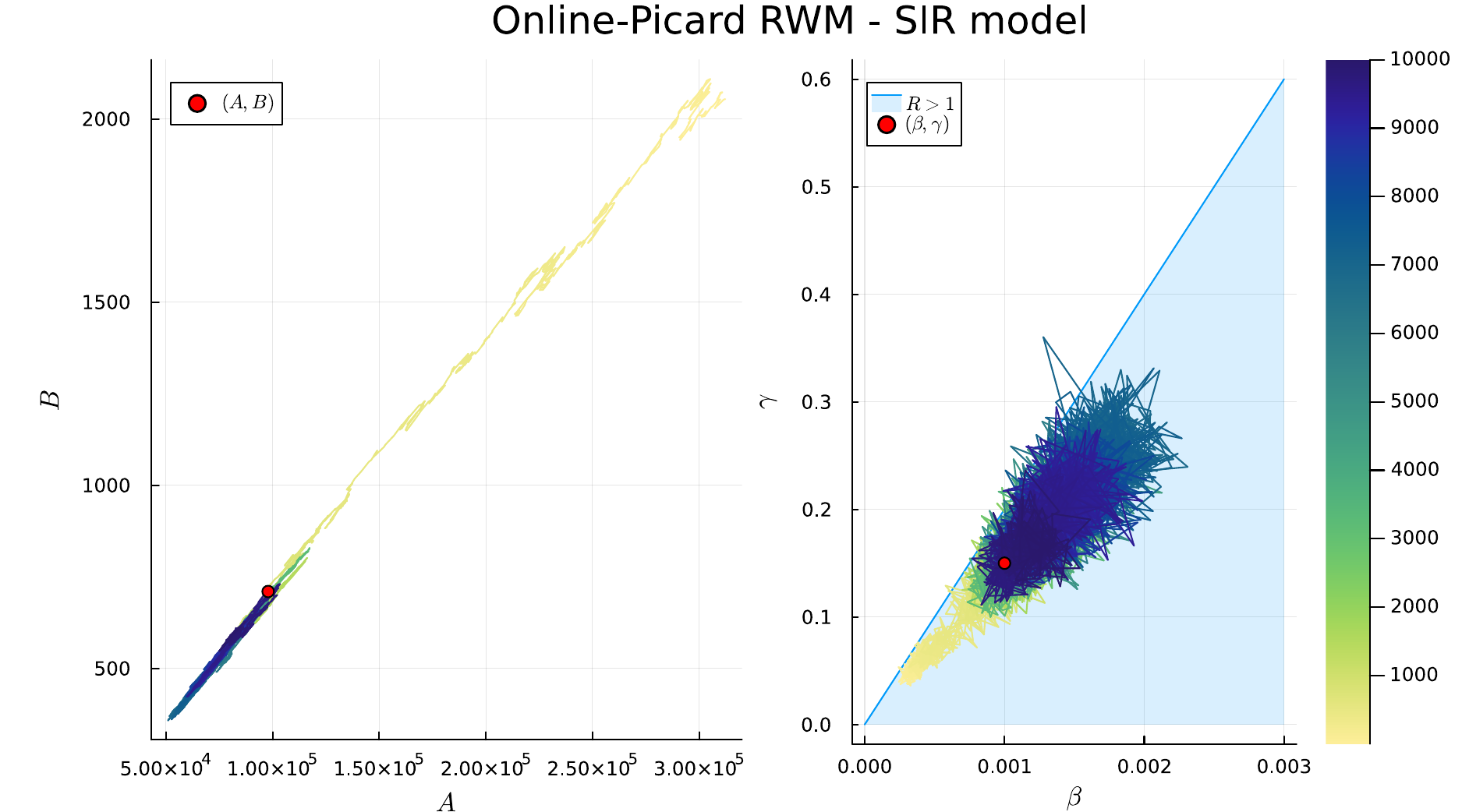}
    \caption{ $(A(X_i, x^{\circ}), B(X_i, x^{\circ}))$ (left panel) and independent samples of $(\gamma, \beta) \mid (X_{10 \,i}, x^{\circ})$ (right panel), from $i = 0$ (yellow) to $i = 10^{5}$ (blue) and $X_i$ being the $i$th step of the RWM. The regions $R_0 \ge 1$ is displayed in sky-blue. The true values from which we simulated the data are shown with red points.
    \label{fig: sir}}
\end{figure}

\subsection{A real-world application for zeroth-order parallel algorithms}\label{sec:blackbox}
We apply here a parallel implementation of the online Picard algorithm on a real statistical application in precision medicine. In this application, each patient measurement is modeled as a complex delayed ordinary differential equation, and the goal is to infer the model parameters to understand the effect of  personalized treatments on advanced-stage
cancer patients, see \citet{jenner2021silico} for more details. 
This application provides a concrete example where, due to black-box log-target evaluations where gradients are not available, \emph{zeroth-order} MCMC methods are a natural choice. This example is also well-suited to parallel methods, given that each target evaluation is computationally expensive, as it requires running numerical solvers.

A full description of the model and dataset is provided in \cite[][Section~5.2]{gagnon2023improving}, where it was also used as a test-case for parallel zeroth-order MCMC methods.  The parameter dimension is $d = 14$. On an Apple M3 CPU with $16$GB of RAM, a single evaluation of $\pi$ takes approximately $0.25$ seconds, while one iteration of a parallel implementation of the Picard algorithm  with $K=8$ parallel CPU cores 
takes approximately $ 0.4$ seconds, see Section~\ref{sec: practical algorithmic aspects} in the supplement for details. We set the total number of iterations to $N = 10^5$. 
For this application, we observe $\hat G = 4.37$ and an effective speedup (i.e.\ reduction in total wall-clock time, see Section \ref{sec:effective speed-up}) of 2.52 when using the parallel implementation of the online Picard algorithm rather than the sequential implementation of RWM. These results are consistent with \eqref{eq: observed speedup} in Section~\ref{sec:effective speed-up} and the theoretical results in Section~\ref{sec: Complexity of Picard algorithms}.

\section{Extensions}
We outline here some extensions of this work, which are left for future research.
First, the Online Picard recursion \eqref{eq: online picard 1} sets $X^{(j+1)}_i$ for $i \ge U^{(j)}$ to a constant value equal to the last coordinate of the Picard map $\Phi$. While this choice is simple and reasonable, it does not exploit the intrinsic dynamics of the underlying Markov chain and the noise variables beyond the time index $U^{(j)}$. More elaborate recursions are possible, combining Picard maps with inexpensive predictive functions that estimate the values of the Markov chain at those steps where the Picard map is not been applied yet. An example of such schemes is provided by the (stochastic) Parareal algorithm \citep{lions2001resolution, pentland2022stochastic}. It would be interesting to explore and characterize the gains obtained by these approaches in our high-dimensional zeroth-order context.
Second, by analyzing the family of approximate measures $\pi_r$ for $r > 0$ in Section~\ref{sec: approximate online-picard}, it may be possible to establish the parallel round complexity of the Approximate Online Picard algorithm. Our numerical results suggest that this algorithm can achieve a parallel round complexity of order $\cO(1)$ using $\cO(d)$ parallel processors.
Finally, Proposition~\ref{prop: MC picard sup} in the supplementary material suggests that trying to characterize explicitly the invariant distribution of  the Markov chain $(X^{(j)}_{L^{(j)}:L^{(j)}+K}, W_{L^{(j)}:L^{(j)}+K-1})_{j \in \NN}$ might be a promising way to improve the analysis of the Online Picard algorithm, e.g.\ going beyond classical log-concavity or obtaining more explicit expressions, or quantifying the error made by the approximate Picard version.

\section*{Acknowledgments}
Both authors acknowledge support from the European Research Council (ERC), through the Starting grant ‘PrSc-HDBayLe’, project number 101076564.  Both authors are grateful to Florian Maire for guiding us through the black-box code used for the application in Section~\ref{sec:blackbox}.

\bibliographystyle{biometrika}
\bibliography{paper-ref}

\pagebreak
\appendix

{
\begin{center}
\textbf{\Large Supplementary Materials for “Parallel computations for Metropolis Markov chains with Picard maps” by S. Grazzi and G. Zanella}
\end{center}
}
\section{Assumption~\ref{ass: 1} for Bayesian Logistic Regression}\label{sec:logistic_regression}
Here we verify Assumption \ref{ass: 1} for the logistic regression with Gaussian prior $\cN_d(0, \Lambda^{-1}_0)$. Assume for $i =1,2,\dots,n$, $Y_i \mid A_i \overset{\text{i.i.d.}}{\sim} \mbox{Ber}(\psi(\langle A_i, x \rangle )),$ with $\psi(u) =(1 + e^{- u})^{-1}$ and unknown parameter $x \in \RR^d$, response variables $Y_i \in \{0,1\}$ and covariates $A_i \in \RR^d$. Then, the negative log-posterior takes the form 
$$
V(x) = \frac{1}{2}x^\top\Lambda_0 x + \sum_{i=0}^n \left(y_i \log(\psi(A_i^\top x)) + (1- y_i) \log(1- \psi(A_i^\top x))\right)
$$
with Hessian 
$$
\nabla^2 V(x) = \Lambda_0 + \sum_{i=1}^n \psi(A_i^\top x)(1 -\psi(A_i^\top x))A_i A_i^\top. 
$$
Since  $\psi(A_i^\top x)(1 -\psi(A_i^\top x))\le 1/4$, we have that 
$$
\nabla^2 V(x) \preccurlyeq \Lambda_0 + \frac{1}{4}\sum_{i=1}^n  A_i A_i^\top,
$$
thus $V$ is $L$-smooth with 
$$
L := \lambda_{max}(\Lambda_0 +  \frac{1}{4}\sum_{i=1}^n  A_i A_i^\top).
$$
Define $w(t) := \psi(t)(1 - \psi( t))$. One can check that $w$ is Lipschitz with constant $C = 1/(6\sqrt{3})$. Since $\nabla^2 V(y) - \nabla^2 V(x) = \sum_{i=1}^n (w(A_i^\top y) - w(A_i^\top x))A_i A_i^\top$, we have
$$
- \|y-x\| M \preccurlyeq  \nabla^2 V(y) - \nabla^2 V(x) \preccurlyeq  \|y-x\| M
$$
with $M =  C \sum_{i=1}^n \|A_i\|A_i A_i^\top$. Thus, Assumption~\ref{ass: 1} holds with 
$
\gamma =  \mathrm{Tr}(M) / (L^{3/2} d^{1/2}),
$
and $\mathrm{Tr}(M) = C \sum_{i=1}^n \|A_i\|^3$. Note that, if one further assumes that $A_i \sim \cN(0, I_d / d)$ (as in Section~\ref{sec: high-dimensional regressions}), $\Lambda_0 = \alpha I_d$ and $n/d \to \rho< \infty$, then, by standard high-probability bounds for Gaussians, one can easily show that 
$\mathrm{Tr}(M) = n C (1 + o_{\PP}(1))$, 
$L = \alpha + \frac{1}{4}(1 + \sqrt{\rho})^2 + o_{\PP}(1)$,
where $X_d = o_{\PP}(a_d)$ means $X_d/a_d \to 0$ in probability. Thus
$$
\gamma = C \frac{n}{L^{3/2} d^{1/2}}(1 + o_{\PP}(1))\,.
$$

\section{Proofs for Random Walk Metropolis}\label{app: proofs}

\subsection{Background lemmas}
We first state some background lemmas that will be helpful to streamline the proofs of the main results.

\begin{lemma}\label{lemma: bound exponential}
    Let $X \sim \cE(\lambda)$, where $\cE(\lambda)$ denotes the exponential distribution with rate $\lambda$. 
    Then for any $a,b \in \RR$ we have $\PP(X \in [a \wedge b, a \vee b]) \le \PP(X \le a \vee b - a \wedge b)$, where $a\wedge b$ and $a\vee b$ denote, respectively, the minimum and maximum between $a$ and $b$.
\end{lemma}
\begin{proof} Without loss of generality, let $a \le b$. For the cases $a,b \le 0$ and $a \le 0, b \ge 0$, Lemma~\ref{lemma: bound exponential} is trivial, while for $a,b \ge 0$, it
 follows from the fact that the density function of an exponential random variable is non-increasing.
\end{proof}
\begin{lemma}\label{lemma: mini-lemma 2}
$Z \sim \cN(0, \sigma^2I_d)$ with $\sigma\geq 0$. 
Then, for every $v\in\RR^d$, we have
$
 \EE[|\langle v, Z\rangle|
]=\|v\|\sqrt{\frac{2}{\pi}}\sigma.
$
\end{lemma}
\begin{proof}
    By rotational invariance of $\cN(0, \sigma^2I_d)$, we have $
     \EE[|\langle v, Z\rangle| =\|v\| \EE[| Z_{1}|]
$,
where $|Z_1|$ follows a one-dimensional folded normal distribution and has mean $\sqrt{\frac{2}{\pi}}\sigma$.
\end{proof}

\begin{lemma}\label{lemma: quadratic symmetric matrix}
Let $Z\sim \cN(0, \sigma^2 I_d)$ and $M$ be a $d \times d$ random matrix independent of $Z$. Then   
$\EE[Z^\top M Z] = \sigma^2 \mathrm{Tr}(\EE[M])$.
\end{lemma}
\begin{proof}
$\EE[Z^\top M Z]= \sum_{i,j}\EE[ Z_i Z_j] \EE[M_{i,j}]= \sum_{i}\EE[ Z_i^2] \EE [M_{i,i}] =  \sigma^2 \mathrm{Tr}(\EE[M])$.
\end{proof}

\begin{lemma}\label{lemma: bounds expectation with indicator function}
     For any non-negative random variable $Z$ and constant $c > 0$, we have 
     $$
     \EE[Z \ind (Z > c)] =  \int_c^\infty  \PP(Z \ge u)] \dd u  + c \PP(Z > c).
     $$
\end{lemma}
\begin{proof}
Follows from
\begin{align*}
\EE[Z \ind (Z > c)]-c \PP(Z > c)
&= \EE[ \int_0^Z  \ind ( u > c)\dd u ] = \EE[ \int_c^\infty  \ind ( u < Z)\dd u]\\
&=\int_c^\infty  \EE[\ind ( u < Z)\dd u]= \int_c^\infty  \PP(Z \ge u) \dd u\,,
\end{align*}
where we used Fubini-Tonelli's Theorem in the third equality. 
\end{proof}

\subsection{Contraction of the Picard map: key lemmas}
The following three lemmas are at the core of the analysis of the contraction of the Picard map $\Phi$, as well as of the proofs of the main results.

\begin{lemma}\label{lemma: 1}
    Under Assumptions~\ref{ass: 1}-\ref{ass: 1.2}, for every $x, y \in \cX$,
    $$
    \PP(f(x, W) \ne f(y, W)) \le \frac{h L^{1/2}}{d^{1/2}}\left(\sqrt{\frac{2}{\pi}} + \frac{h \gamma}{2}\right)\|x-y\|, \qquad W\sim \nu.
    $$
\end{lemma}

\begin{proof}
 \label{app: proof: lemma 1}
By \eqref{eq:zero_f}, under Assumptions \ref{ass: 1}-\ref{ass: 1.2}, we have that
$$
\PP(f(x, W) \ne f(y, W)) = \PP(B(x, U, Z) \ne B(y, U, Z))
$$
with $U \sim \text{Unif}([0,1])$ and $Z \sim \cN(0, h^2/(Ld)I_d)$. 
Since $E = -\log(U) \sim \cE(1)$, by Lemma~\ref{lemma: bound exponential} we have
\begin{align*}
    \PP(B(x,  U,Z)&\ne B(y,  U,Z) \mid Z )\\
    &=\left|\PP\left(E < V(x + Z) - V(x)\right) - \PP\left( E < V(y + Z) - V(y) \right) \right|\\
    &\le 1 - \exp(-|\Delta| )
\end{align*}
where $\Delta =  V(x + Z) - V(x) - (V(y+ Z) - V(y))$. Thus
\begin{align*}
\PP(B(x,  U,Z) \ne B(y,  U,Z))
&=   \EE[\PP(B(x,  U,Z) \ne B(y,  U,Z) \mid Z )]
\\&\leq
\EE[
1 - e^{-|\Delta|}
] 
\\&\leq 
\EE[
|\Delta|
]\,,
\end{align*}
since $1 - e^{-x} \le x, x \in \RR$.

Since $V$ is twice-continuously differentiable we have
\[
\Delta =  \underbrace{(\nabla V(y) - \nabla V(x))^\top Z}_{\Delta_1} + \underbrace{\frac{1}{2}Z^\top (\nabla^2 V(x + tZ) - \nabla^2 V(y + tZ)) Z}_{\Delta_2}
\]
for some $t=t(x,y) \in [0, 1]$.
By the triangle inequality
\begin{equation*}
    \EE[|\Delta|] \le \EE[|\Delta_1|] +\EE[|\Delta_2|].
\end{equation*}
For the first term, by 
$Z \sim \cN(0, h^2/(Ld)I_d)$, Lemma~\ref{lemma: mini-lemma 2} and the $L$-smoothness of $V$, we have
$$
\EE[|\Delta_1|] = \sqrt{\frac{2}{\pi}}\frac{h}{\sqrt{L d}}\|\nabla V(y) - \nabla V(x)\| \leq \sqrt{\frac{2}{\pi}}h\sqrt{\frac{L}{d}}\|x-y\|^2.
$$
For the second term, by the strong Hessian-Lipschitz condition in Assumption \ref{ass: 1} and Lemma~\ref{lemma: quadratic symmetric matrix}, we have
\begin{align*}
    \EE[|\Delta_2|] 
    &= \frac{1}{2}\EE[|Z^\top (\nabla^2 V(x + tZ) - \nabla^2 V(y + tZ)) Z |] \\
    &\le \frac{1}{2}\EE[|Z^\top M Z]\|x-y\|\\
    &= \frac{h^2}{2}\frac{\mathrm{Tr}(M)}{Ld}\|x - y\|\\
    &\le\frac{h^2}{2}\frac{\gamma L^{3/2} d^{1/2}}{Ld}\|x - y\|\\
    &=\frac{h^2\gamma L^{1/2}}{2d^{1/2}}\|x - y\|\,.
\end{align*}
Hence, we have
\begin{align*}
    \EE[|\Delta|] &\le \EE[|\Delta_1|] + \EE[|\Delta_2|]
    \le \frac{h L^{1/2}}{d^{1/2}}\left(\sqrt{\frac{2}{\pi}} + \frac{h \gamma}{2}\right) \|x-y\|
\end{align*}
as desired.
\end{proof}

\begin{lemma}
    \label{lemma: 2}
    Under Assumption~\ref{ass: 1.2} and for all $x,y \in \cX^{K+1}$ with $ x_0 = y_0$, $w_0\in \cW$ and $1 < i \le d$,
    \begin{align}\label{eq: lemma 6}
        \EE[\|\Phi_i(x,  W) - \Phi_i( y, W)\|^2 \mid W_0 = w_0] &\le \frac{15 h^2}{L}\sum_{\ell = 1}^{i-1}\left(\PP(f(x_\ell, W_\ell) \ne f(  y_\ell,  W_\ell)) + \delta(d)\right),
    \end{align}
with $W_{1:K-1} \sim \nu^{\otimes (K-1)}$, $\nu^{\otimes p}$ being product measure of $\nu$ repeated $p$ times, and $\delta(t)$ as in Theorem~\ref{thm: main theorem 1}. The inequality \eqref{eq: lemma 6} holds also when $x$ and $y$ are replaced by random variables $X, Y$ on $\cX^{K+1}$, which can depend on $W$. 
\end{lemma}

The proof of Lemma~\ref{lemma: 2} relies on the following bounds.

\begin{lemma}\label{lemma:wish}
Let $W\sim \hbox{Wishart}(I_i,d)$ be a Wishart matrix on $\RR^{i\times i}$ with $d$ degrees of freedom and identity scale matrix, and $Y$ be a random vector on $\RR^i$ defined on the same probability space, with $\|Y\|^2\leq i$ almost surely. Then
\begin{equation*}
\EE[Y^TWY]\leq 15\EE[\|Y\|^2]
+i25 \exp(-3d/2)\,.
\end{equation*}
\end{lemma}
\begin{proof}

By \citet[Theorem 6.1]{wainwright2019high}, for all $t \ge 0, \, i \le d$, 
$$
\PP\left(\lambda_{\max}(W) > (1 + \sqrt{i/d} + t)^2 \right) \le e^{-t^2 d/2}
$$
for its maximum eigenvalue. Combining this result and Lemma~\ref{lemma: bounds expectation with indicator function}, we have 
\begin{align*}
  \EE[\lambda_{\max}(W) \ind(\lambda_{\max}(W) \ge 15)]   &=
  \int_{15}^\infty \PP(\lambda_{\max}(W) \ge u)\, \dd u +  15 \PP(\lambda_{\max}(W) > 15) \\
  &\le \int_{15}^\infty  e^{- (\sqrt{u}-2)^2 d/2} \dd u + 15 e^{-d (\sqrt{15} - 2)^2/2} \\
&\le \int_{15}^\infty  e^{- u d/10} \dd u + 15 e^{-3d /2}\\
  &= \frac{10}{d}\exp(-3d/2) + 15 \exp(-3d/2)\\
  &\le 25 \exp(-3d/2)
\end{align*}
where we used $(\sqrt{u} - 2)^2 > u/5$ for $u \ge 15$.
Finally, 
\begin{align*}
    \EE[Y^T W Y]&\le \EE[\lambda_{max}(W) \ind(\lambda_{max}(W) \le 15)\|Y\|^2] + i\EE[\lambda_{max}(W) \ind(\lambda_{max}(W) > 15)]) \\
    & \le 15 \left(\EE[\|Y\|^2] + i\delta(d)\right)
\end{align*}
with $\delta(d) = \frac{5}{3} \exp(-3d/2)$.
\end{proof}

\begin{proof}[of Lemma~\ref{lemma: 2}]
By the definition of $\Phi_i$ in \eqref{eq: picard recursion 0} and $f$ in \eqref{eq:zero_f}, we have 
\begin{align} \label{eq: norm picard map}
\Phi_i(x,  W) - \Phi_i(y,  W)
&=
\sum_{\ell=1}^{i -1}\Delta B_\ell Z_\ell 
\end{align}
where $\Delta B_\ell =  B( x_\ell, U_\ell, Z_\ell)- B( y_\ell, U_\ell, Z_\ell)$, and \eqref{eq: norm picard map} is independent of $W_0$. 
Here $U_1,U_2,\dots,U_{i-1} \overset{\mathrm{i.i.d}}{\sim}\text{Unif}([0,1])$ and $Z_1,Z_2,\dots,Z_{i-1} \overset{\text{i.i.d}}{\sim} \cN(0, h^2 /(Ld))$.
Also,
\begin{align}
\label{eq: IW2}
\|\sum_{\ell=1}^{i -1}\Delta B_\ell Z_\ell\|^2 =   \sum_{\ell,j =1}^{i-1}\Delta B_\ell \Delta B_j \langle Z_\ell, Z_j\rangle
 \stackrel{d}=
\frac{h^2}{L} \Delta B^T W \Delta B
\end{align}
with  $\Delta B = (\Delta B_1,\dots,\Delta B_{i-1})^T$ and $W \sim \hbox{Wishart}(I_{i-1}, d)$. Here $ \stackrel{d}=$ denotes equality in distribution. 
Combining \eqref{eq: norm picard map}-\eqref{eq: IW2} and using Lemma~\ref{lemma:wish} we have
\begin{align*}
&\EE[\|\Phi_i(x,  W) - \Phi_i(y,  W)\|^2\mid W_0=w_0]=
\frac{h^2}{L} \EE[ \Delta B^T W \Delta B]  \le \frac{15h^2}{L} (\EE[\|\Delta B\|^2] + i\delta(d))
\end{align*}
 with $\delta(d) =  \frac{5}{3}\exp(-3d/2)$.
The result follow by $\EE[\|\Delta B\|^2]=\sum_{\ell=1}^{i -1}\PP( B(x_\ell, U_\ell, Z_\ell)\ne B( y_\ell, U_\ell, Z_\ell))$. 
\end{proof}

Fix $i\in\{0,\dots,d\}$ and define 
\begin{align*}
A^{(j)} &=  \max_{\ell \le i}\PP(f(X^{(j)}_\ell,  W_\ell) \ne f(X_\ell, W_\ell)),
&j\in\{0,1,\dots\}\,.    
\end{align*}
\begin{lemma}
\label{lemma: recursion for max probability}
Under the assumptions of Theorem~\ref{thm: main theorem 1}, we have
    \begin{equation*} 
    A^{(j+1)} \le \sqrt{c_0\frac{i}{d}\left(A^{(j)} + \delta(d)\right)}
    \end{equation*}
    with $c_0$, $\delta(d)$ as in Theorem~\ref{thm: main theorem 1}.
\end{lemma}

\begin{proof}
By Lemma~\ref{lemma: 1} and Jensen's inequality we have
    \begin{align*}
    \PP(f(X^{(j+1)}_i,  W_i) \ne f(X_i, W_i)) & = \EE[\PP(f(X^{(j+1)}_i,  W_i) \ne f(X_i, W_i)\mid X^{(j+1)}_i, X_i)]\\
    &\le  \frac{h L^{1/2}}{d^{1/2}}\left( \sqrt{\frac{2}{\pi}} + \frac{h \gamma}{2} \right)\EE[\| X^{(j+1)}_i-X_i\|]
    \\
    &\le  \frac{h L^{1/2}}{d^{1/2}}\left( \sqrt{\frac{2}{\pi}} + \frac{h \gamma}{2} \right)\sqrt{\EE[\| X^{(j+1)}_i-X_i\|^2]}\,.
    \end{align*}
    By Lemma~\ref{lemma: 2} we have that
    \begin{align*}
    \EE[\| X^{(j+1)}_i-X_i\|^2]&=\EE[\| \Phi_i(X^{(j)},W)-\Phi_i(X,W)\|^2]\\
    &\le\frac{15h^2}{L} \sum_{\ell = 1}^{i-1}   \left(\PP(f(X^{(j)}_\ell,  W_\ell) \ne  f(X_\ell, W_\ell)) + \delta(d)\right)\\
    &\le\frac{15h^2}{L} i   \left(A^{(j)} + \delta(d)\right)\,.
    \end{align*}
Combining the above gives
    \begin{align*}
    \PP(f(X^{(j+1)}_i,  W_i)\ne f(X_i, W_i))  
    &\le  \sqrt{15} h^2\left( \sqrt{\frac{2}{\pi}} + \frac{h \gamma}{2} \right) \sqrt{\frac{i}{d}\left(
    A^{(j)} + \delta(d)\right)}\,,
    \end{align*}
which implies the desired result after noting that $\max_{\ell\leq i}$ is a non decreasing operator in $i$.
\end{proof}

\subsection{Proof of Theorem~\ref{thm: main theorem 1}}\label{app: proof: thm 1}

\begin{lemma}\label{lemma:recursion}
Let $(a_j)_{j=0,1,\dots}$ be a non-negative sequence satisfying $a_0 = 1$ and
$a_{j+1}\leq b\sqrt{a_j+\epsilon}$ for some fixed $b,\epsilon>0$ and all $j\geq 0$. Then
$a_j\leq b^2+\epsilon+2^{-j}$ for all $j\geq 0$.
\end{lemma}
\begin{proof}
Let $g(a)=b\sqrt{a+\epsilon}$.
By monotonicity of $g$, we can assume without loss of generality that $a_0=1$ and $a_{j+1}= g(a_j)$ for all $j\geq 1$. 
The function $g$ is concave and its unique fixed point on $[0,\infty)$ satisfying $a^*=g(a^*)$ is 
$$
a^*=\frac{b^2+\sqrt{b^4+4\epsilon b^2}}{2}\in(b^2,b^2+\epsilon)\,.
$$
Indeed, imposing $g(a)^2=a^2$ gives $a^2-b^2a-b^2\epsilon=0$, whose solutions are $(b^2+\sqrt{b^4+4\epsilon b^2})/2$ and $(b^2-\sqrt{b^4+4\epsilon b^2})/2$, and the latter is negative.
Also, $(b^2+\sqrt{b^4+4\epsilon b^2})/2\geq (b^2+\sqrt{b^4})/2=b^2$ and, by concavity of $x\mapsto \sqrt{x}$, $\sqrt{b^4+4\epsilon b^2}\leq b^2+\frac{1}{2\sqrt{b^4}}4\epsilon b^2=b^2+2\epsilon$ which implies
$a^*=(b^2+\sqrt{b^4+4\epsilon b^2})/2\leq (b^2+b^2+2\epsilon)/2=b^2+\epsilon$. 

If $a_0=1\leq a^*$ then $a_j\leq a^*\leq b^2+\epsilon\leq b^2+\epsilon+2^{-j}$ for all $j\geq 0$, as desired.
Consider now $a_0=1\geq a^*$.
Since $g'$ is strictly decreasing on $[0,\infty)$, for all $a\geq a^*$ we have $0< g'(a)\leq g'(a^*)\leq g'(b^2-\epsilon)=1/2$.
Thus, $g$ is a strict contraction on $[a^*,\infty)$ with Lipschitz constant smaller than $1/2$ and, by Banach's fixed point theorem, we have
$
|a_j-a^*|\leq 2^{-j}|a_0-a^*|
$ for all $j\geq 0$.
By $a_0=1\geq a^*$, we have $|a_0-a^*|\leq a_0=1$ and thus we deduce
$
a_j
\leq
a^*+ 2^{-j}
\leq
b^2+\epsilon+ 2^{-j}
$
as desired.
\end{proof}

\begin{proof}[of Theorem~\ref{thm: main theorem 1}].
    Follows directly by Lemma \ref{lemma: recursion for max probability} and Lemma \ref{lemma:recursion} with $b=\sqrt{c_0 i/d}$, $\epsilon=\delta(d)$, and $c_0, \delta(d)$ as in Theorem~\ref{thm: main theorem 1}.
\end{proof}

\subsection{Proof of Corollary~\ref{corol: picard map}}\label{app: proof: thm 2} 
By \eqref{eq: online picard 2} we have
   \begin{align*}
   \PP(L^{(j)} < K) &=
    1-\PP \left( \bigcap_{i=0}^{K -1 } f(X^{(j)}_i,  W_i) = f(X_i, W_i) \right)\\ 
    &= \PP \left(\bigcup_{i=0}^{K-1} f(X^{(j)}_i,  W_i) \ne f(X_i, W_i) \right)\,.
\end{align*}
Applying the union bound and the bound in Theorem~\ref{thm: main theorem 1}, for $K \le b \sqrt{d}, b > 0$, it follows
   \begin{align*}
   \PP(L^{(j)} < K) &\leq \sum_{i=0}^{K-1} \PP(f(X^{(j)}_i,  W_i)\ne f(X_i, W_i))
    \le  b\sqrt{d}\left(\frac{c_0 b}{\sqrt{d}} + \delta(d) + 2^{-j}\right)\,.
\end{align*}
This is less than $\epsilon$ if
$c_0 b^2 = \frac{\epsilon}{2}$  and 
$b \sqrt{d}(\delta(d) + 2^{-j}) \le \epsilon/2$.
The first condition is achieved by setting $b = \frac{\epsilon}{\sqrt{2 c_0}}$. With this choice, the latter condition is always satisfied for $h\ge 1, \, d\ge 1$
since $(h,d) \mapsto \sqrt{\frac{d}{2c_0}}(\delta(d) + 2^{-j}) - \frac{1}{2}$ is negative on $[1,\infty)^2$.

\subsection{Proof of Theorem~\ref{thm: complexity rwm}}
\label{app: proof compelxity rwm}
For two real-valued random variables $X \sim \mu,\, Y \sim \nu$, we write $X\succeq Y$ if $\mu([x,\infty)) \ge \nu([x, \infty))$ for all $x\in(-\infty,\infty)$.

\begin{proposition}
\label{prop: Gi}
Under Assumption \ref{ass: 1}-\ref{ass: 1.2},  $h \ge 1$, for all $\epsilon > 0$, $K\leq \sqrt{\epsilon^2 d/(2 c_0)}$, 
 $x_0 \in \cX$, any $j \in \NN$ and $p = \lceil \log(d) \rceil$, we have 
    $$
    (L^{(jp)} \mid X_0 = x_0) \succeq \sum_{i=1}^j P_i,
    $$
    with $P_i \overset{i.i.d.}{\sim} (1 - \epsilon)\delta_K(\cdot) +  \epsilon\delta_{0}( \cdot)$.
\end{proposition}
 The proof of Proposition~\ref{prop: Gi} requires the following two lemmas. 
For any $i \in \NN$, define $\tilde{G}_i := \sum_{\ell=0}^{p-1} G^{(p i + \ell)}$, with $(G^{(i)})_{i=0,1,\dots}$ defined in \eqref{eq:def_G_j} below, $\delta_c(\cdot)$ being the Dirac measure centered at $c$.
\begin{lemma}\label{lemma: G1}
    Suppose that, for some random variable $A$, $\tilde G_i \mid \tilde G_{0:(i-1)}, A  \succeq P_i$. Then $\tilde G_i \mid \tilde G_{0:(i-1)} \succeq P_i$.
\end{lemma}
\begin{proof}
    By the law of total probability we have 
    $$
    \PP(\tilde G_i \ge t \mid \tilde  G_{0:(i-1)}) = \EE_A[\PP(\tilde G_i \ge t \mid G_{0:(i-1)}, A)] \ge \EE_A[\PP(P_i \ge t)] = \PP(P_i \ge t).
    $$
\end{proof}

\begin{lemma}
\label{lemma: stoch ordering for sum of rvs} Suppose $\tilde G_i \mid \tilde G_{0:(i-1)} \succeq P_i$ for all $i\ge 0$. Then $\sum_{i = 0}^j \tilde G_i \succeq \sum_{i = 0}^j P_i$ for all $j \ge 0$.
\end{lemma}
\begin{proof} Let $\hat G_j := \sum_{i=0}^j \tilde G_i$, $\hat P_j := \sum_{i=0}^j P_i$. Then
 $\hat G_0 = \tilde G_0 \succeq P_0 = \hat P_0$. By induction, if $\hat G_{j} \succeq \hat P_j $, for some $j$, then we have 
\begin{align*}
    F_{\hat G_{j+1}}(t) =F_{   \tilde G_{j+1} + \hat G_{j}
    }(t)&=\int F_{ \tilde G_{j+1} \mid \hat G_{j}}(t-u)\dd F_{\hat G_{j}}(u) \\
    &\le \int F_{P_{i+1}}(t-u)\dd F_{\hat G_{j}}(u) \\
    &\le \int F_{P_{i+1}}(t-u)\dd F_{\hat P_{j}}(u) = F_{\hat P_{j+1}}(t) 
\end{align*}
where $F_X(t)$ denotes the distribution function of $X$ evaluated at $t$ and the last inequality follows since $u \mapsto F_{P_{i+1}}(t-u)$ is decreasing.
\end{proof}

\begin{proof}[of Proposition~\ref{prop: Gi}]
As also shown in the proof of Proposition~\ref{prop: MC picard sup} below, we have 
$
L^{(j+1)} = \sum_{i=0}^j G^{(i)}
$. Let $\tilde L_i := L^{(pi)}=\sum_{\ell=0}^{i-1} \tilde{G}_\ell$. Note that $
\{\tilde L_i = l\} \in \sigma(X_0, W_{0:l})$ where $\sigma(X)$ denotes the sigma-algebra generated by the random variable $X$. 

Let $A = (X_0, W_{0:\tilde L_{i}}, \tilde L_i)$. Then $\tilde G_{0:(i-1)} \in \sigma (A)$, $$\PP(\tilde{G}_i \ge \ell \mid A, \tilde{G}_{0:(i-1)}) = \PP(\tilde{G}_i \ge \ell \mid A)$$
and, for any $a = (x_0, w_{0: l}, l)$, $\epsilon > 0$ and $K \le \epsilon \sqrt{d/(2c_0)}$ we have
\begin{align*}
 \PP(\tilde{G}_i \ge \ell \mid A = a) &=  \PP(\tilde{G}_i \ge \ell \mid X_{\tilde L_i} = x_{l}, W_{\tilde L_i} = w_{l}, \tilde L_i = l) \\
 &= \PP(\tilde{G}_0 \ge \ell \mid X_{\tilde L_0} = x_{l}, W_{\tilde L_0} = w_{l}) \\
 &=\PP(L^{(p)} \ge \ell \mid X_{0} =  x_{l}, \, W_{0} = w_{l})\\
 &\geq  \PP(P_i \ge \ell )\,,
\end{align*} 
where $x_{i} = x_0 + \sum_{\ell=1}^{i-1} f(x_\ell, w_\ell), \, i \ge 1$ and the last inequality follows from Corollary~\ref{corol: picard map} and from $\tilde{G}_i \ge p$ almost surely, which holds by construction. 
Combining the above with Lemmas~\ref{lemma: G1}-\ref{lemma: stoch ordering for sum of rvs} yields the desire result.
 \end{proof}

\begin{proof}[of Theorem~\ref{thm: complexity rwm}]
    Fix $\epsilon = 1/2$. By Proposition~\ref{prop: Gi}, for all $j, N \in \NN$, $p = \lceil \log(d)\rceil$, $K \le \sqrt{d/(8c_0)}$, we have
\begin{align*}
    \PP(L^{(pj)} < N \mid X_0 = x)  &= \PP(\sum_{i=0}^{j-1} G^{(i, p)} < N\mid X_0 = x) \le \PP\left(\sum_{i=1}^j P_i < N\right)\,
\end{align*}
with $(P_i)_{i=1}^j$ as in Proposition~\ref{prop: Gi}, $\mu = \EE[P_1] =  (1-\epsilon)K =  K/2$.
By Hoeffding's inequality, we have 
\begin{align*}
    \PP\left(\sum_{i=1}^j P_i< N\right)
    &= \PP\left(\sum_{i=1}^j P_i - j\mu < -(j \mu - N)\right)\\
    &\le \exp\left(- \frac{2(j\mu - N)^2}{j K^2}\right)\\
    &= \exp\left(- \frac{2(jK/2 - N)^2}{j K^2}\right) \le\exp\left(- j\left(1-\frac{2N}{j K}\right)\right)
\end{align*}
for any $j > 2N/K$. 
In particular, for any $\delta > 0$, if 
$j \geq \max\{\frac{4N}{K},2\log(1/\delta)\}$, then $j\left(1-\frac{2N}{j K}\right)\geq 
2\log(1/\delta)\left(1-\frac{1}{2}\right)=\log(1/\delta)$ and the above exponential bound is less than $\delta$.
\end{proof}

\subsection{Proof of Corollary~\ref{corollary: complexity rwm}}
Let $(X_n)_{n=0,1,2,\dots}$ be the RWM trajectory, i.e.\ the exact solution of the recursion in \eqref{eq: 1}. 
Under Assumptions~\ref{ass: 1}-\ref{ass: 1.2}, $m I_d \preccurlyeq \nabla^2 V$, $h=1$ and 
$X_0 \sim \cN(x^\star, L^{-1}I_d)$, we can apply Theorem 49 and Remark 50 of \citet{andrieu2024explicit} to deduce $\sqrt{\chi^2 (\cL(X_{n_{mix}}), \pi)} \le \epsilon$ with $n_{mix}=\mathcal{O}(\frac{L}{m}d)$, where $n_{mix}$ depends on $\epsilon$, $d$, $L/m$ and is explicitly provided in \citet[Rmk.50]{andrieu2024explicit}. 
Consider running the Online Picard algorithm with $K = \sqrt{d/(8 c_0)}$ for 
$$J=\lceil \log(d)
    \max\{4n_{mix}/K,2\log(2/\epsilon)\}$$ 
parallel iterations and returning $Y=X^{(J)}_{\min\{n_{mix},L^{(J)}\}}$ as output.
When $n_{mix}\leq L^{(J)}$, we have $Y=X^{(J)}_{n_{mix}}=X_{n_{mix}}$.
By Theorem~\ref{thm: complexity rwm}, $\PP(Y=X_{n_{mix}})\geq \PP(L^{(J)}\geq n_{mix}) \ge 1-\epsilon/2$. Therefore
\begin{align*}
\|\cL(Y)-\pi\|_{TV}
&\leq
\PP(Y\neq X_{n_{mix}})+TV(\cL(X_{n_{mix}}), \pi)\\
&\leq \epsilon/2+\frac{1}{2}\sqrt{\chi^2 (\cL(X_{n_{mix}}), \pi)}\\
&\leq \epsilon/2+\epsilon/2=\epsilon\,.
\end{align*}

\subsection{Proof of Proposition~\ref{prop: transient phase}}\label{app: proof prop transient phase}

For a sequence $x_n \to \infty$, let $A_n=V(x_n+Z)-V(x_n)$, $B_n=V(Y_n +Z) - V(Y_n)$, $Z \sim \cN(0,1)$, where $Y_n = x_n + N_n$ and for a sequence of random variables $(N_n)_{n=1,2,\dots}$ on $\RR^d$ with $\EE[\|N_n\|^2] \le C$.
\begin{lemma}
    \label{lem: 1 prop tails}
    Under Assumption~\ref{ass: 1} with $mI_d \preccurlyeq \nabla^2 V$, we have
    \begin{align}
       \EE[|A_n - B_n|] &\le 2L(C + d), &\hbox{for all } n = 1,2,\dots\,, \label{eq: lem prop tails 1}\\
       \limsup_{n}\PP(|A_n| \le c) &= 0 &\hbox{for all } n = 1,2,\dots\,, \, c \in (0, \infty). \label{eq: lem prop tails 2}
    \end{align}    
\end{lemma}
\begin{proof}
Let $A = V(x + z) - V(x)$, $B =V(y + z) - V(y)$, for some $x, y, z \in \cX$. By $L$-smoothness of $V$, we have that, for some $t_1,t_2\in(0,1)$ depending on $x, y, z$,
\begin{align*}
|A-B|
&=
|\langle \nabla V(x+t_1 z)-\nabla V(y+t_2z),z\rangle|\\
&\leq \|\nabla V(x + t_1 z)-\nabla V(y + t_2 z)\|\|z\|\\
&\leq L(\|x-y\|+\|z\|)\|z\|\\
&\leq 2L\|x-y\|^2+\|z\|^2.
\end{align*}
Thus, 
\begin{align*}
    \EE[|A_n - B_n|]&\leq 2L\EE[\|x_n-Y_n\|^2+\|Z\|^2] \\
    &=  2L(\EE[\|N_n\|^2]+\EE[\|Z\|^2])\\
    &\le 2L(C + d) 
\end{align*}
proving \eqref{eq: lem prop tails 1}.
Also, again by $L$-smoothness of $V$, 
\begin{align*}
|A_n| \geq 
|\langle \nabla V(x_n),Z\rangle|-L\|Z\|^2/2\,
\end{align*}
thus, for all $c \in (0, \infty)$
\begin{align}
    \label{eq: lemma 1 prop tails}
   \PP(|A_n|\leq c)\leq 
\PP(\|Z\|^2\geq \frac{2}{L}(|\langle \nabla V(x_n),Z\rangle|-c)).
\end{align}
By rotational invariance $Z$, the right hand-side of \eqref{eq: lemma 1 prop tails} depends on $\nabla V(x_n)$ only through its norm. By $m$-convexity of $V$, $\|\nabla V(x_n)\| \to \infty$  as $|x_n| \to \infty$ and, consequently, the right hand-side of \eqref{eq: lemma 1 prop tails} goes to 0, proving the desired result in \eqref{eq: lem prop tails 2}.
\end{proof}

\begin{lemma}
\label{lem: 2 prop tails}
Under Assumption~\ref{ass: 1} with $mI_d \preccurlyeq \nabla^2 V$, we have
\begin{align*}
\limsup_{n\to\infty}
\EE[| \exp(-\max\{A_n,0\} ) - \exp(-\max\{B_n,0\})|]=0\,.
\end{align*}
\end{lemma}
\begin{proof} 
For a fixed $c \in (0,\infty)$, we have that
\begin{multline*}
    \EE[| \exp(-\max\{A_n,0\} ) - \exp(-\max\{B_n,0\})|]
\\\leq
e^{-c}+\PP(|A_n|\leq c \cap |B_n|\leq c) +\PP(|A_n - B_n| > 2c).
\end{multline*}
By Lemma~\ref{lem: 1 prop tails}, $\limsup_{n\to\infty}\PP(A_n>c \cap B_n > c) = 0$ and
$$
\PP(|A_n - B_n| > 2c)\le \frac{\EE[|A_n - B_n|]}{2c}\le   \frac{L(C + d)}{c}.
$$
Combining these results and since $c$ was arbitrary, we have
    \begin{align*}
&\limsup_{n\to\infty}\EE[| \exp(-\max\{A_n,0\} ) - \exp(-\max\{B,0\})|]
&\le \inf_c\{e^{-c} + \frac{L(C + d)}{c}\} = 0\,.
\end{align*}
\end{proof}

\begin{proof}[of Proposition~\ref{prop: transient phase}] 
Here we prove $\lim_{n \to \infty} \PP(L^{(j)} = jK) = 1$ for $j = 1$. The case $j \ge 1$ follows by recursion.
By \eqref{eq:zero_f}, we have
    \begin{align*}
        1 - \PP(L^{(1)} = K) &= 
        1 - \PP\left( \bigcap_{i=0}^{K} X^{(1)}_i = X_i\right)\\
        &= \PP\left(\bigcup_{i=0}^{K} \left( X^{(1)}_i\ne  X_i\right) \right)\\
        &=\PP\left(\bigcup_{i=0}^{K-1} \left( B(X^{(0)}_i, W_i) \ne  B(X_i, W_i) \right)\right)\\  
        &\le\sum_{i=0}^{K-1} \PP \left( B(x_n, W_i) \ne  B(X_i, W_i) \right).
    \end{align*}
Hence, $\lim_{n \to \infty} \PP(L^{(1)} = K) = 1$, if $\lim_{n \to \infty}\PP \left( B(x_n, W_i) \ne  B(X_i, W_i) \right) = 0$ for all $i = 0,1,\dots,K-1$.

For a fixed $i$, note that $X_i = x_n + N_n$
where $N_n = \sum_{\ell= 0}^{i-1}B(X_{\ell}, U_\ell, Z_\ell) Z_\ell$ and
\begin{align*}
    \EE[\|N_n\|^2] &= \EE[\|\sum_{\ell = 0}^{i-1}B(X_{\ell}, U_\ell, Z_\ell) Z_\ell\|^2]\\
    &\le \sum_{\ell = 0}^{i-1}\EE[\|Z_\ell\|^2] + \sum_{\substack{\ell,m = 0\\ \ell\ne m}}^{i-1}\EE[\|Z_\ell\|]\EE[\|Z_m\|]  \le C,
\end{align*}
with $C = i^2 d$. Then, for $E = -\log(U_i) \sim \mathrm{Exp}(1)$,
\begin{align*}
    &\PP ( B(x_n, U_i,  Z_i) \ne  B(X_{i},  U_i,  Z_i) )\\ 
    &=  \EE[\PP\left( B(x_n, U_i,  Z_i) \ne  B(x_n + N_n,  U_i,  Z_i) \mid U_{0:i-1},\, Z_{0:i} \right)]\\
    &=\EE [\left|\PP\left(E < V(x_n + Z) - V(x_n)\right) - \PP\left( E < V(x_n + N_n + Z) - V(x_n + N_n) \right) \right|]\\
    &=   \EE[| e^{-\max(V(x_n + Z) - V(x_n),0)} - e^{-\max(V(x_n + N_n + Z) - V(x_n + N_n),0)}|] \\
    &=   \EE[| e^{-\max(A_n,0)}- e^{-\max(B_n,0)}|] 
\end{align*}
which goes to 0 as  $n \to \infty$ by Lemma~\ref{lem: 2 prop tails}.
\end{proof}

\subsection{Proof of Proposition~\ref{prop: approximate picard}}\label{app: corollary 1}
 By Markov inequality and the bound in Theorem~\ref{thm: main theorem 1}, we have for $K \le \lfloor b d \rfloor$, $b = \epsilon r/(3c_0)$,
\begin{align*}
    \PP\Big(\frac{1}{K}\sum_{i=0}^{K-1}\ind \big(f(X_i^{(j)}, W_i) \ne f(X_i, W_i)\big) \le r\Big) &\le \frac{1}{Kr}\sum_{i=0}^{K-1}\EE\left[\ind \left(f(X_i^{(j)}, W_i) \ne f(X_i
    , W_i)\right)\right]\\
    &\le \frac{\sum_{i=0}^{\lfloor bd\rfloor }\PP(f(X^{(j)}_i,  W_i) \ne f(X_i, W_i))}{\lfloor bd \rfloor r} \\
    &\le \frac{c_0 b + \delta(d) + 2^{-j}}{r}\le \epsilon
\end{align*}
for $d \ge -2\log(\epsilon r/5)/3$, $j \ge -\log(\epsilon r/3)/\log(2)$.

\section{Proofs for Metropolis within Gibbs}\label{app: proofs of MwG}
\subsection{Proof of Theorem~\ref{thm: gibbs}}
Theorem~\ref{thm: gibbs}  essentially shows that most of the results derived in Section~\ref{sec: main results} holds (with different constants) also for Metropolis within Gibbs (MwG) Markov chains, i.e. when replacing Assumption~\ref{ass: 1.2} with Assumption~\ref{ass: MwG}. 

Lemmas \ref{lemma: 1}-\ref{lemma: 2} are the building blocks of all the results of Section~\ref{sec: main results}. Thus, we start by providing the analogue of these Lemmas for MwG.

\begin{lemma}\label{lemma 1 mwg}
    Under Assumptions \ref{ass: 1} and \ref{ass: MwG},  the statement of Lemma~\ref{lemma: 1} holds.
\end{lemma}
To prove Lemma~\ref{lemma 1 mwg}, we first need the following proposition.
\begin{proposition}
    \label{prop: MwG kernel}
    Let $W^{\mathrm{MwG}}$ and $W^{\mathrm{RWM}}$ be the innovations of MwG and RWM satisfying  respectively Assumption~\ref{ass: MwG} and Assumption~\ref{ass: 1}. Then, for every $x \in \cX$
    \begin{align*}
           \cL(f_i(x, W^{\mathrm{MwG}})) &=  \cL(f(x, W^{\mathrm{RWM}})) &\hbox{for } i = 0,1,\dots,d-1,
    \end{align*}
    with $f$ as in \eqref{eq:zero_f} and $f_i$ as in \eqref{eq: MwG}.
\end{proposition}
\begin{proof}
By \eqref{eq:zero_f} and \eqref{eq: MwG},
 Proposition~\ref{prop: MwG kernel} follows if $\cL(o_i Z^{\mathrm{MwG}}) = \cL(Z^{\mathrm{RWM}})$ for all $i$, where $Z^{\mathrm{MwG}} = h/\sqrt{L d} (2P - 1)S$, $P \sim \mathrm{Ber}(1/2)$, $S \sim \chi(d)$, $Z^{\mathrm{RWM}}  \sim \cN(0, I_d h^2/(Ld))$, $(o_0,o_1,\dots,o_{d-1}) \sim \mathrm{Haar}(d)$. By rotational invariance of $(o_0,o_1,\dots,o_{d-1})$ and $Z^{\mathrm{RWM}}$, it is enough to show that
 $$\cL(\|o_i Z^{\mathrm{MwG}}\|) = \cL(\|Z^{\mathrm{RWM}}\|) = \cL(h/\sqrt{L d} \chi(d)).
 $$
\end{proof}

\begin{proof}[of Lemma~\ref{lemma 1 mwg}]
    By Proposition~\ref{prop: MwG kernel}, Lemma~\ref{lemma: 1} follows also under Assumptions~\ref{ass: 1}-\ref{ass: MwG}, with exactly the same constants and proof provided in Section~\ref{app: proof: lemma 1}. 
\end{proof}
Next, we provide the analogue of Lemma~\ref{lemma: 2} for MwG.
\begin{lemma}\label{lemma 2 mwg}
    Under Assumption \ref{ass: MwG} and for all $x,y \in \cX^{K+1}$, with $x_0 = y_0$, $w_0 \in \cW$ and $1 < i \le d$
    \begin{align}
        \label{eq: lemma 6 for MWG}
        \EE[\|\Phi_i(x,  W) - \Phi_i( y, W)\|^2 \mid W_0 = w] &\le \frac{2 h^2}{L}\sum_{\ell = 1}^{i-1}\left(\PP(f_i(x_\ell, W_\ell) \ne f_i( y_\ell,  W_\ell)) + \delta(d)\right)
    \end{align}
    with $\delta(d) = 11 \exp(-d/10)$ and where the expectation is relative to $(W_1,W_2,\dots,W_{K-1})$.
\end{lemma}
The proof of Lemma~\ref{lemma 2 mwg} relies on the following bounds.
\begin{lemma}\label{lem:chi_bern}
Let $X$ and $Y$ be two random variables defined on the same probability space, with $X = \tilde X/d, \,  \tilde X \sim \chi^2_d$ and $Y$ binary, i.e.\ $\PP(Y\in\{0,1\})=1$.
Then
$$
\EE[X Y]\leq 2\left(\PP(Y=1)+ 11\exp(-d/10)\right) \, .
$$
\end{lemma}
\begin{proof}
By \citet[Lemma~1]{laurent2000adaptive}, for every $t \ge 2$,
\begin{equation*}
    \PP(X > t) \le e^{-td/20}.
\end{equation*}
Then, Combining this result and Lemma~\ref{lemma: bounds expectation with indicator function}, we have that
\begin{align*}
  \EE[\ind(X \ge 2)X] &=
  \int_2^\infty \PP(X \ge u)\, \dd u + 2 \PP(X > 2) \\
  &\le \int_2^\infty  e^{- ud/20} \dd u + 2 e^{-d/10}\\ 
  &= \frac{20}{d}\exp(-d/10) + 2 \exp(-d/10) \\
  &\le 22 \exp(-d/10).
\end{align*}
Finally
\begin{align*}
    \EE[Y X] &\le \EE[ Y X \ind(X  \le 2 )] + \EE[(X \ind (X  > 2)]\\
    &\le 2\left(\EE[Y] + \delta(d)\right)
\end{align*}
with $\delta(d) = 11 \exp(-d/10)$.
\end{proof}

\begin{proof}[of Lemma \ref{lemma 2 mwg}]

Under Assumptions~\ref{ass: 1} and \ref{ass: MwG}, 
\begin{align*}
    \Phi_i(x,  W) - \Phi_i(y,  W)
&=
\sum_{\ell=1}^{i -1}\Delta B_\ell Z_\ell o_\ell 
\end{align*}
where $\Delta B_\ell =   B( x_\ell, U_\ell, o_\ell Z_\ell) - B( y_\ell, U_\ell, o_\ell Z_\ell)$, $Z_\ell = h/\sqrt{L d} (2P_\ell - 1)S$, $U_1,U_2,\dots,U_{i-1} \overset{\mathrm{i.i.d}}{\sim}\text{Unif}([0,1])$,  $P_1,P_2,\dots,P_{i-1} \overset{\mathrm{i.i.d}}{\sim} \mathrm{Ber}(1/2)$ and $S_1,S_2\dots,S_{i-1} \overset{\mathrm{i.i.d}}{\sim} \chi(d)$.

Since $\langle o_i,o_j\rangle = 0$ almost surely for all $i \ne j$ and by Proposition~\ref{prop: MwG kernel}, we have
$$
\|\sum_{\ell=1}^{i -1}\Delta B_\ell o_\ell Z_\ell\|^2 =   \sum_{\ell,j =1}^{i-1}\Delta B_\ell \Delta B_j \langle o_\ell Z_\ell, o_j Z_j\rangle =  \sum_{\ell =1}^{i-1}\Delta B_\ell^2 \|o_\ell Z_\ell\|^2 = \frac{h^2}{L}\sum_{\ell =1}^{i-1}\Delta B_\ell^2 \tilde Z_\ell 
$$
with $\tilde Z_1,\, \tilde Z_2, \dots, \tilde Z_{i-1}  \overset{\mathrm{i.i.d}}{\sim} \chi^2(d)$. Thus, by Lemma \ref{lem:chi_bern}, we have

\begin{align*}
    \EE[\|\sum_{\ell=0}^{i -1}\Delta B_\ell o_\ell Z_\ell\|^2] &= 
    \frac{h^2}{L} \sum_{\ell=1}^{i -1}\EE[\Delta B_\ell^2\tilde Z] \nonumber\\
    &\le \frac{2 h^2}{L}\sum_{\ell=1}^{i -1}\left(\PP( B(x_\ell, U_\ell, o_\ell Z_\ell)\ne B( y_\ell, U_\ell, o_\ell Z_\ell)) + \delta(d)\right) 
\end{align*}
with $\delta(d) = 11 \exp(-d/10)$.
\end{proof}

\begin{proof}[of Theorem~\ref{thm: gibbs}]
Under Assumptions~\ref{ass: 1} and \ref{ass: MwG}, the statements of Lemma~\ref{lemma: recursion for max probability} and Theorem~\ref{thm: main theorem 1} hold with new constants 
\begin{align}\label{eq: consants MwG}
    c_0 = 2h^4 \left(\sqrt{\frac{2}{\pi}} + \frac{h\gamma}{2} \right)^2, \qquad \delta(d) = 11 \exp(-d/10).
\end{align}
The proofs follow the same steps of the proof of  Lemma~\ref{lemma: recursion for max probability} and Theorem~\ref{thm: main theorem 1} in Section~\ref{app: proof: thm 1}.

Similarly, Under Assumptions~\ref{ass: 1} and \ref{ass: MwG}, the statements of Corollary~\ref{corol: picard map} and Theorem~\ref{thm: complexity rwm} hold 
with constants $c_0, \delta(d)$ as in \eqref{eq: consants MwG}, $K = \epsilon \sqrt{\frac{d}{2c_0}}$ and  $h \ge 5$. The proofs follow the same steps of the proof of  Corollary~\ref{corol: picard map} in Section~\ref{app: proof compelxity rwm} and the proof of Theorem~\ref{thm: complexity rwm} in Section~\ref{app: proof: thm 2}.

Finally,  Under Assumptions~\ref{ass: 1} and \ref{ass: MwG} the statement of Proposition~\ref{prop: approximate picard}  holds with $c_0$ as in \eqref{eq: consants MwG} and $d \ge -10\log(\epsilon r/33)$.
\end{proof}

\subsection{Proof of Proposition~\ref{prop: instant convergence ORWM isotropic Gaussians}\label{app: proof prop orwm}}
Let $\pi(x) \propto \exp(-\|x\|^2/(2\sigma^2))$ and let  $X^{(0)}_i = X_0$ for $i =0,1,\dots,K$. We prove here that $L^{(1)} = K$, for all $X_0 = x_0 \in \cX$, $K \le d$. Proposition~\ref{prop: instant convergence ORWM isotropic Gaussians}  then  follows by recursion.

For MwG, $L^{(1)} {=} K$ if, for all $w = (w_0,w_1,\dots,w_{k-1})\in \cW^{K}$,
\begin{align}\label{eq: mwg gauss 1}
   f_i(x_0, w_i) &{=} f_i(x_i, w_i) &\hbox{for all }0\le i < K   
\end{align}
where $(x_0,x_1,x_2\dots,x_K)$ is defined by the recursion $x_{i+1} = x_i + f_i(x_i, w_i)$, $f_i$  as in \eqref{eq: MwG}.

By \eqref{eq: MwG}, for all $u = (u_0,u_1,\dots,u_{K-1}) \in [0,1]^{K}$, $(z_0,z_1,\dots,z_{K-1}) \in \RR^{K}$, the condition in \eqref{eq: mwg gauss 1} is equivalent to
\begin{align}\label{eq: B prop}
 B(x_0,  u_i,  \hat z_i) &{=} B(x_i, u_i,  \hat z_i)&\hbox{for all }0\le i < K 
 \end{align}
 with $  \hat z_i = o_i  z_i$ and a given orthonormal basis $o_0,o_1,\dots,o_{d-1}$, $d \ge K$. Then, the left-hand side of \eqref{eq: B prop} is
\begin{align}
    B(x_i,  u_i,  \hat z_i) &=  B(x_0 + \sum_{\ell = 0}^{i-1}B(x_\ell,  u_\ell, \hat z_\ell) \hat z_\ell,  u_i,  \hat z_i) \nonumber \\
    &= \ind \left(\frac{\pi(x_0 + \sum_{\ell = 0}^{i-1}B(x_\ell, u_\ell,  \hat z_\ell) \hat z_\ell +  \hat z_i)}{\pi(x_0 + \sum_{\ell = 0}^{i-1}B(x_\ell,  u_\ell,  \hat z_\ell)  \hat z_\ell)} > u_i\right) \nonumber \\
    &=  \ind\left(\frac{\| \hat z_i\|^2}{2\sigma^2} +  \frac{\langle x_0 + \sum_{\ell = 0}^{i-1}B(x_\ell, u_\ell,  \hat z_\ell) \hat z_\ell,  \hat z_i\rangle}{\sigma^2} > u_i\right),\label{eq: prop eq 1}
\end{align}
while the right-hand side of \eqref{eq: B prop} is
\begin{align}
    B(x_0,  u_i, z_i)
    &=  \ind\left(\frac{\pi(x_0 +   \hat z_i)}{\pi(x_0)}> u_i\right) \nonumber\\
    &=  \ind\left(\frac{\| \hat z_i\|^2}{2\sigma^2} + \frac{\langle x_0,  \hat z_i \rangle}{\sigma^2}> u_i\right). \label{eq: prop eq 2}
\end{align}
Because $o_0,o_1,\dots,o_{d-1}$ is an orthonormal basis on $\RR^d$ we have that 
\begin{align*}
    \langle  \hat z_\ell,  \hat z_i\rangle &= z_\ell z_i\langle o_\ell ,  o_i\rangle {=} 0 &  \hbox{for all } i \ne \ell,
\end{align*}
hence, \eqref{eq: prop eq 1} and \eqref{eq: prop eq 2} are equal.
\section{Efficient Implementation of the (Approximate) Online Picard Algorithm}\label{app: pseudo-code}

Our efficient implementation is based on reformulating the Picard algorithm as a Markov chain on $K$-dimensional vectors. 

Let $(X^{(j)}, L^{(j)})_{j=0,1,\dots}$ be the sequence of random variables produced by the Online Picard recursion of \eqref{eq: online picard 1}-\eqref{eq: online picard 2}. Define $\bar{X}^{(j)}=X^{(j)}_{L^{(j)}:U^{(j)}}$ and $\bar{W}^{(j)}= W_{L^{(j)}:U^{(j)}-1}$.
\begin{proposition}\label{prop: MC picard sup}
   The sequence $(\bar{X}^{(j)},\bar{W}^{(j)})_{j=0,1,\dots}$
    is a time-homogeneous Markov chain on $(\cX^{K+1}\times \cW^K)$.
\end{proposition}

\begin{proof}
The sequence $(\bar{X}^{(j)},\bar{W}^{(j)})_{j=0,1,\dots}$ can be described by the recursion 
  \begin{align}
    \label{eq: mc online picard 1}
    \bar X^{(j+1)}_{i} &= \begin{cases}
        \Phi_{G^{(j)} + i}(\bar X^{(j)}, \bar W^{(j)}) 
    & i \le K-G^{(j)}\,,\\
     \Phi_{K}(\bar X^{(j)}\,, \bar W^{(j)}) 
    & i > K-G^{(j)}\,,\\
    \end{cases}  & i = 0,1,\dots,K\,,\\
    \label{eq: mc online picard 2}
    \bar W^{(j+1)}_i &= \begin{cases}
    \bar W^{(j)}_{G^{(j)}+i} & i \le K-G^{(j)}-1\,,\\
    \tilde W_{i,j} &  i > K-G^{(j)}-1\,, 
    \end{cases}  & i = 0,1,\dots,K-1\,,
  \end{align}
  where $\tilde W_{i,j}, \, i, j \ge 0,$ are i.i.d. random variables  with distribution $\nu$ on $\cW$,
  \begin{equation}\label{eq:def_G_j}
        G^{(j)}= G(\bar X^{(j)},\bar W^{(j)}) = \sup\{i \le K \colon \Phi_\ell(\bar X^{(j)}, \bar W^{(j)}) =  \bar X^{(j)}_\ell \text{ for }  0 \le \ell\le i\}\,,
  \end{equation}
and \eqref{eq: mc online picard 1}-\eqref{eq: mc online picard 2} define the one-step Markov transition kernel of $(\bar{X}^{(j)},\bar{W}^{(j)})_{j=0,1,\dots}$.  
\end{proof}
Proposition~\ref{prop: MC picard sup} sheds lights to the underlying (Markov) structure of Online Picard recursion and is used in  Algorithm~\ref{alg: OPA sleek} for a memory-wise and computationally efficient implementation  of the algorithm that iterates over a vector of length $K$. The Approximate Online Picard algorithm is identical to Algorithm~\ref{alg: OPA sleek} except for line 5, which is replaced by
\[
G = \sup\{1 \le i \le K \colon \bar \cA_\ell \le r \text{ for all } \ell \le i\}, 
\]
where
\[
\bar \cA_\ell = \frac{|\{0 \le s < \ell \colon f(\bar X_{s},  \bar W_{s}) \ne f(\bar X^c_{s},  \bar W_{s})\}|}{\ell}\,.
\]

\begin{algorithm}[!h]
\caption{Online Picard algorithm}  \label{alg: OPA sleek}
\KwIn{$N, K \in \NN$, $X_{0} \in \cX$.}
 Initialize $\bar X^{c}_i = \bar X_i = X_0$, for $i = 0,1,\dots,K$\;
 Set $L =0$\;
 Set $\bar W_0,\bar W_1,\dots,\bar W_{K-1} \overset{\text{i.i.d}}{\sim}\nu$ \;
\While{$L < N$}{
$\bar X = \Phi(\bar X^c, \bar W)$\;
   $G = \sup\{i \le K \colon \bar X_\ell =  \bar X^{c}_\ell \text{ for }  0 \le \ell\le i\}$\;
   $\bar X^c_0 = \bar X_G$ \;
   \For{$i = 0,1,\dots, K-1$}{
   \eIf{$i \le K - G - 1$}{
    $\bar X^{c}_{i+1} =\bar X_{G+i+1}$\;
    $\bar W_i = \bar W_{G + i}$\;
   }{
   $\bar X^{c}_{i+1} = \bar X_K$\;
   $\bar W_i \sim \nu $ \;
   }
   }
   $L = L + G$\;
   }
\KwOut{$\bar X_{G - (L-N)}$.}
\end{algorithm}

\section{Numerical simulations for MwG}\label{sup: simulations mwg}
 We run  the same numerical experiments as in Section~\ref{sec: high-dimensional regressions} for the (Approximate) Online Picard algorithms applied to Metropolis within Gibbs was applied. Figure~\ref{fig: mwg} summaries the results. $d, N, K$ and $r$ are as in Section~\ref{sec: high-dimensional regressions}.

\begin{figure}[!ht]
    \centering
    \includegraphics[width=0.8\linewidth] {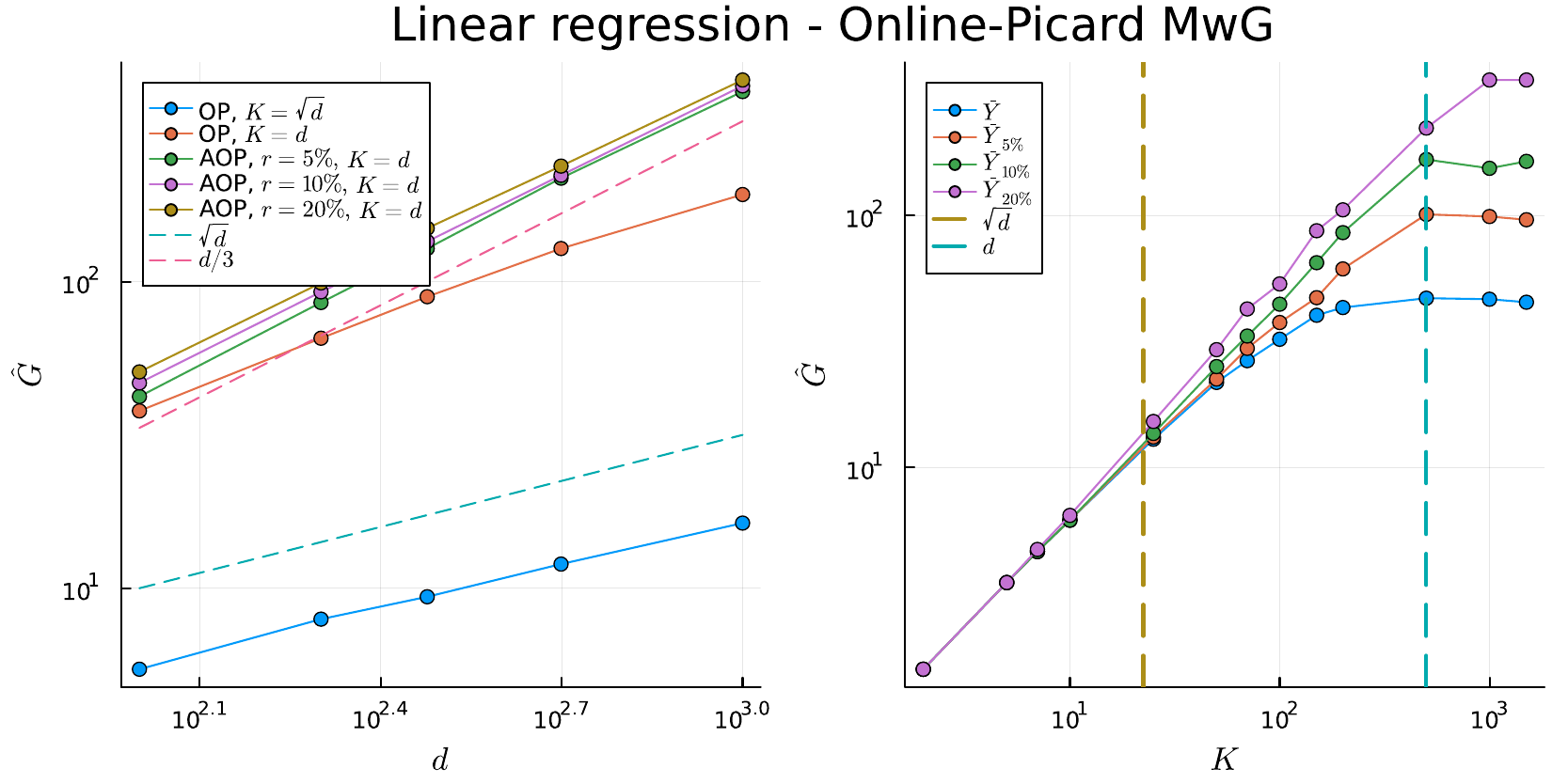}
    \includegraphics[width=0.8\linewidth] {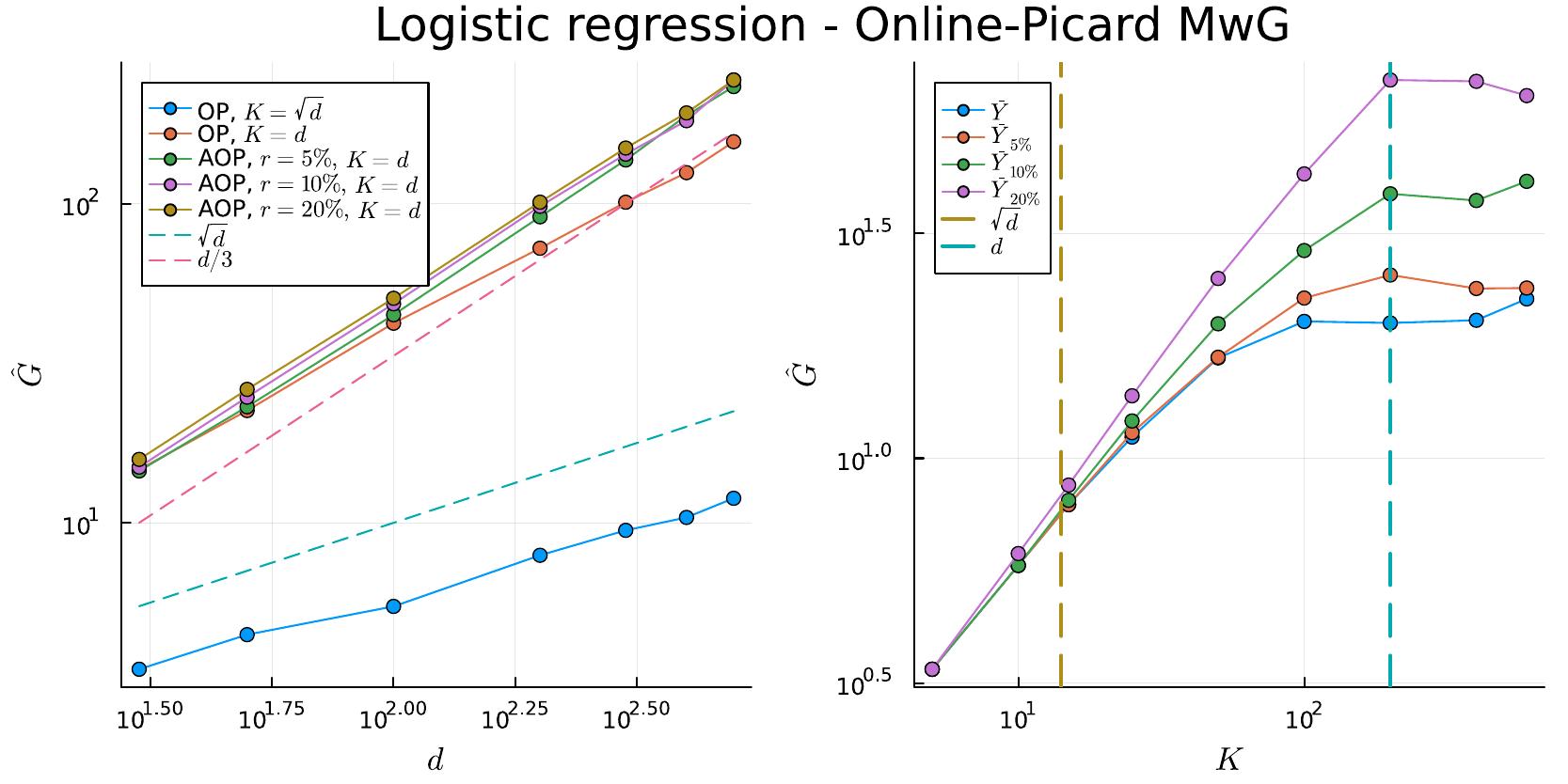}
    \includegraphics[width=0.8\linewidth] {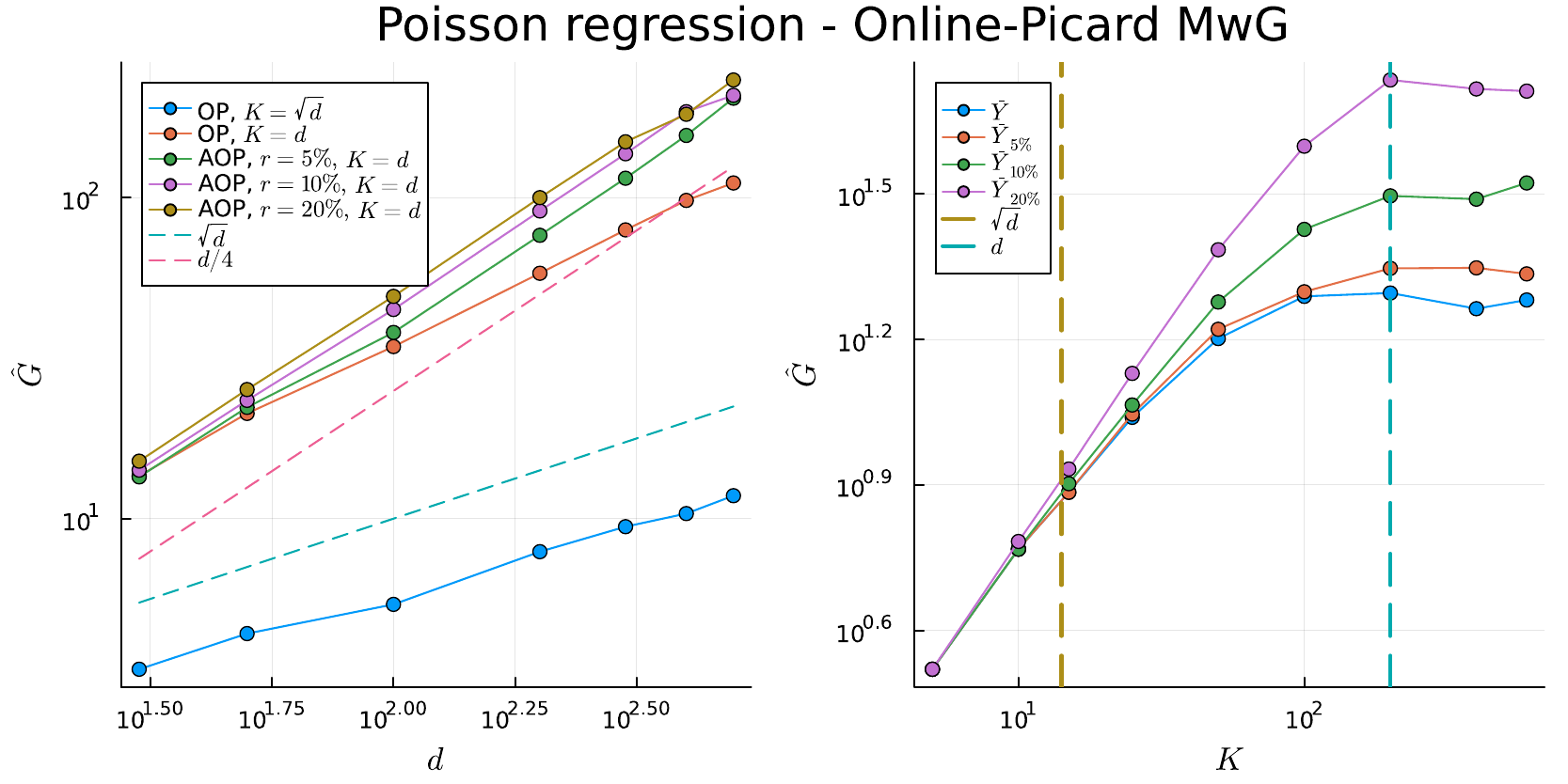}
    \caption{Performance of Online Picard algorithm ($\bar X$) and its approximate versions ($\bar X_r, \, r =5\%,\dots,20\%$) applied to MwG, with target being the linear regression model \textbf{E1} (top panels), logistic regression model \textbf{E2} (middle panels) and Poisson regression model \textbf{E3} (bottom panels).} 
    \label{fig: mwg}
\end{figure}

\section{Definition of Effective Sample Size}\label{sec: ESS}
The Effective Sample Size (ESS) of a test function $f \colon \cX \to \RR$ is defined as
\begin{equation}
    \label{eq: ess}
    \text{ESS}(f) = \frac{1}{1 + 2 \sum_{i=1}^\infty \rho_i(f)}
\end{equation}
where $\rho_i(f)$ is the lag-$i$ autocorrelation function of $\{f(X_i)\}_{i=1,2,\dots}$. 
The Online Picard algorithm simulates a trajectory $(X_0,\dots,X_N)$ of the original Markov chain with $N/\hat{G}$ parallel iterations, resulting in an increase of ESS relative to the original sequential implementation of 
$\hat G$.

\section{Practical algorithmic aspects}\label{sec: practical algorithmic aspects}
\subsection{Parallel Architecture of Picard algorithms}\label{sec: parallel architectures}
Here we discuss how the main processor interacts with the $K$ parallel processors at each iteration of the Picard map in \eqref{eq: picard recursion 0}. Our discussion focuses on RWM, but the same considerations apply to MwG.

Recall that, by Assumption~\ref{ass: 1.2}, a RWM increment is given by 
\begin{equation}\label{eq: rwm increment}
    f(x, z, u) = \ind\left(\frac{\pi(x+z)}{\pi(x)}> u\right).
\end{equation}
Therefore, a naive implementation of the Picard map in \eqref{eq: picard recursion 0} requires each processor to evaluate $\pi$ twice to compute the increment. This contrasts with an efficient sequential implementation of RWM, which requires a single evaluation of $\pi$ per iteration (Algorithm~\ref{alg:sequential rwm}).

\begin{algorithm}[!h]
\caption{One iteration of the sequential implementation of RWM\label{alg:sequential rwm}}
\KwIn{Position $x\in \RR^d$, evaluation of $\pi(x)$, Gaussian noise $Z$.} 
Compute $\pi(x + Z)$\;
With probability $\min(1, \pi(x + Z)/\pi(x))$ set $x' = x + Z$, otherwise set $x' = x$\;
\KwOut{$x'$, $\pi(x')$.}
\end{algorithm}
However, a careful look at \eqref{eq: picard recursion 0} and \eqref{eq: rwm increment} reveals a more efficient recursion for the Picard map, in which each parallel processor requires only a single evaluation of $\pi$ (Algorithm~\ref{alg: picard map sleek}). \begin{algorithm}[!h]
\caption{Efficient implementation of the Picard map $\Phi$ in \eqref{eq: picard recursion 0} for RWM\label{alg: picard map sleek}}
\KwIn{$K \in \NN$, $X^{(i-1)}_{0:K}$, $W_{0:K-1}$, $\pi(X^{(i-1)}_0)$, $B^{(i-1)}_j := \ind(\pi(X^{(i-1)}_j + Z_j)/\pi(X^{(i-1)}_j) > U_j), \, j = 0,1,\dots K-1$.}
 \textbf{Compute in parallel:} $\pi(X^{(i-1)}_j + Z_j)$ for $j = 0,1,\dots,K-1$\;
 Set $B^{(i)}_0 = \ind(\pi(X^{(i-1)}_0 + Z_0)/\pi(X^{(i-1)}_0) > U_0)$\;
\For{$j = 1,2,\dots, K-1$}{
    \eIf{$B^{^{(i-1)}}_{j-1} = 0$}{
    Set $\pi(X_{j}^{(i-1)}) = \pi(X_{j-1}^{(i-1)})$\;
        }{
        Set $\pi(X_{j}^{(i-1)}) = \pi(X^{(i-1)}_{j} + Z_{j})$\;
        }
        Set $B^{(i)}_j = \ind(\pi(X^{(i-1)}_j + Z_j)/\pi(X^{(i-1)}_j) > U_0)$\;
        Set $X^{(i)}_j = X^{(i)}_{j-1} + B^{(i)}_{j-1}$\;
    }
    Set $X^{(i)}_K = X^{(i)}_{K-1} + B^{(i)}_{K-1}$\;
\KwOut{$\pi(X_0^{(i)})$, $X^{(i)}_{0:K}$ $B^{(i)}_{0:K-1}$.}
\end{algorithm}

Note that, the parallel architecture used in Algorithm~\ref{alg: picard map sleek} is comparable to that used in parallel implementations of Multiple-try schemes \citep{glatt2024parallel}: in both cases, the input is broadcast to $K$ processors, which evaluate $\pi$ at different locations in parallel; see Figure~\ref{fig:parallel_architecture} for an illustration.
\begin{figure}[h!]
    \centering
    \includegraphics[width=0.65\linewidth]{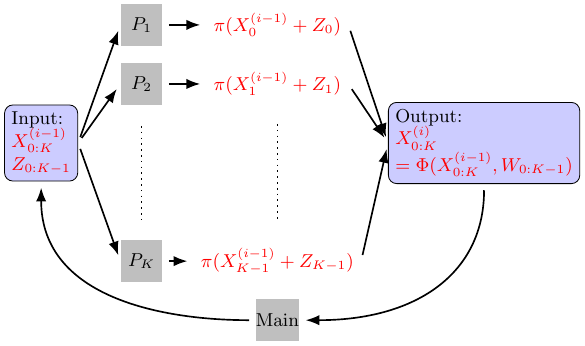}
    \caption{Diagram of one parallel iteration of the Picard map in Algorithm~\ref{alg: picard map sleek}: the input is sent to $K$ parallel processors $P_1,\dots,P_K$, each evaluating $\pi$ at a different location. A similar parallel architecture is used for parallel implementations of Multiple-try.}
    \label{fig:parallel_architecture}
\end{figure}

\subsection{Parallelization overhead and effective speedup}\label{sec:effective speed-up}
To assess the effective speedup of Picard algorithms relative to their sequential counterparts, one must account for the wall-clock time of each algorithm. For a given algorithm, this mainly depends on (i) mixing of the Markov chain, as summarized by ESS in \eqref{eq: ess}; (ii) the computational cost of each evaluation of $\pi$, denoted here by $c$; and, only for parallel algorithms, (iii) the parallelization overhead $\epsilon$ incurred by the parallel architecture used at each iteration. In this section we discuss how these factors affect the observed effective speedup.

We quantify efficiency via ESS per unit wall-clock time, i.e.  $\mathrm{ESS}_{\mathrm{seq}} / T_{\mathrm{seq}}$ and  $\mathrm{ESS}_{\mathrm{pic}} / T_{\mathrm{pic}}$ respectively for the sequential implementation and the online Picard algorithm.
By the discussion in Section~\ref{sec: parallel architectures}, we let $T_{\mathrm{seq}} \approx N \times c$ and $T_{\mathrm{pic}} \approx  M_N \times (c + \epsilon)$, where $M_N$ is the number of parallel Picard iterations needed to simulate $N$ increments of the Markov chain. The \emph{effective speedup} of the Picard algorithm relative to the sequential algorithm is given by:
\begin{equation}\label{eq: observed speedup}
    \frac{\left( \mathrm{ESS}_{\mathrm{pic}}/T_{\mathrm{pic}}\right)}{\left(\mathrm{ESS}_{\mathrm{seq}} /T_{\mathrm{seq}}\right)} = \frac{T_{\mathrm{seq}}}{T_{\mathrm{pic}}} \approx \frac{N}{M_N}\left(1 + \frac{\epsilon}{c}\right)^{-1}  \approx \hat G\left(1 + \frac{\epsilon}{c}\right)^{-1}.
\end{equation}
Here we used $\mathrm{ESS}_{\mathrm{pic}}=\mathrm{ESS}_{\mathrm{seq}}$, since both methods output the same trajectory. For instance, in our application in precision medicine (Section~\ref{sec:blackbox}), $c\approx \epsilon$, i.e., parallelization overhead equals the cost of evaluating $\pi$, and the effective speedup is approximately half the speedup $\hat G$. Note that, by Theorem~\ref{thm: complexity rwm}, we have that $\hat G = C \sqrt{d}$, when $K \ge  \sqrt{d}$, for some $C>0$. Therefore, by \eqref{eq: observed speedup}, achieving an observed speedup of at least $S$ (with $S\le C\sqrt{d}$) requires the parallelization latency $\epsilon$ to be smaller than a fraction $(C\sqrt{d}/S-1)$ of $c$. This condition can be checked for a given application and depends on the dimensionality of the problem, the target-evaluation cost $c$ and the parallelization overhead $\epsilon$, whose value is strongly influenced by the specific parallel architecture and hardware employed. For example, \citet{glatt2024parallel} notes that, for Multiple-try, a GPU implementation with $K = 10^6$ has roughly the same parallelization overhead of an equivalent implementation on a CPU with only $K= 10^2$.

\end{document}